\documentclass[12pt,reqno]{amsproc}

\usepackage[margin=1in]{geometry}

\usepackage[T1]{fontenc}
\usepackage{bbm}
\usepackage[usenames]{color}
\usepackage{mathpazo}
\usepackage{amsmath,amssymb,amsthm}
\usepackage{wrapfig}
\usepackage{tikz-cd}
\usepackage[shortlabels]{enumitem}

\usepackage{caption}
\usepackage{mathrsfs,bm,nicefrac}
\usepackage[all,cmtip]{xy}
\usepackage{subcaption}

\newcommand{\orcid}[1]{\href{https://orcid.org/#1}{\includegraphics[width=10pt]{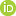}}}

\definecolor{PineGreen}{rgb}{0.0,0.47,0.44}
\definecolor{MidnightBlue}{rgb}{0.1,0.1,0.44}
\definecolor{magenta}{rgb}{1.0,0.0,1.0}
\definecolor{bl1}{HTML}{4479A1}
\definecolor{pur1}{HTML}{52196D}
\definecolor{mag1}{HTML}{2AD0F1}
\definecolor{org1}{rgb}{.92,.39.21}
\definecolor{pur2}{rgb}{.53,.47,.7}
\definecolor{desycyan}{rgb}{0.00,0.68,0.93}
\definecolor{desyorange}{rgb}{0.96,0.52,0.07}
\definecolor{desygray}{rgb}{0.47,0.47,0.47}

\usepackage{hyperref}
\hypersetup{
	colorlinks=true,
	linkcolor=desycyan,
	citecolor=PineGreen,
	filecolor=magenta,
	urlcolor=desyorange
}

\usepackage[labeled,resetlabels]{multibib}
\newcites{P}{Publications}

\makeatletter
\newcommand{\eqnum}{\refstepcounter{equation}\textup{\tagform@{\theequation}}}
\makeatother

\newtheorem{theorem}{Theorem}
\numberwithin{theorem}{section}
\newtheorem{proposition}[theorem]{Proposition}
\newtheorem*{theorem*}{Theorem}
\newtheorem{lemma}[theorem]{Lemma}
\newtheorem{corollary}[theorem]{Corollary}
\theoremstyle{definition}

\newtheorem{definition}[theorem]{Definition}

\theoremstyle{remark}
\newtheorem{remark}[theorem]{Remark}
\newtheorem{example}[theorem]{Example}
\AtEndEnvironment{example}{\null\hfill\ensuremath{\diamond}}
\newtheorem*{example*}{Example}
\AtEndEnvironment{example*}{\null\hfill\ensuremath{\diamond}}

\AtEndEnvironment{remark}{\null\hfill\ensuremath{\spadesuit}}
\newcommand{\Con}{\mathbf{Con}}

\newcommand{\RR}{\mathbb{R}}

\newcommand{\QQ}{\mathbb{Q}}
\newcommand{\PP}{\mathbb{P}}

\newcommand{\pp}{\mathbb{P}}
\newcommand{\CC}{\mathbb{C}}
\newcommand{\ZZ}{\mathbb{Z}}

\newcommand{\scrO}{\mathscr{O}}
\newcommand{\scrF}{\mathscr{F}}
\newcommand{\scrM}{\mathscr{M}}
\newcommand{\scrN}{\mathscr{N}}
\newcommand{\scrJ}{\mathscr{J}}
\newcommand{\scrI}{\mathscr{I}}
\newcommand{\scrD}{\mathscr{D}}

\newcommand{\Eu}{\mathbf{Eu}}

\newcommand{\bV}{\mathbf{V}}

\newcommand{\set}[1]{{\left\{{#1}\right\}}}

\newcommand{\cD}{\mathcal{D}}

\newcommand{\cI}{\mathcal{I}}

\newcommand{\cT}{\mathcal{T}}

\newcommand{\cS}{\mathcal{S}}

\newcommand{\re}{\mathrm{Re}}

\newcommand{\ann}{\mathrm{Ann}}
\newcommand{\Char}{\mathrm{Char}}
\newcommand{\mult}{\mathrm{mult}}

\renewcommand{\ss}{\boldsymbol{s}}

\newcommand{\xs}{.8}

\makeatletter
\DeclareRobustCommand
\myvdots{\vbox{\baselineskip4\p@ \lineskiplimit\z@
		\hbox{.}\hbox{.}\hbox{.}}}
\makeatother

\usepackage{tikz}    
\usetikzlibrary{graphs,graphs.standard,calc,automata}

\usepackage[foot]{amsaddr}

\newcommand{\avg}[1]{\left< #1 \right>} 

\usepackage{tikz-feynman}
\tikzfeynmanset{compat=1.1.0}

\def\DD{D\kern-.7em\raise0.3ex\hbox{\char '55}\kern.33em}

\usepackage{booktabs}
\usepackage{colortbl}
\definecolor{Ftitle}{RGB}{11,46,108}
\colorlet{tableheadcolor}{Ftitle!25} 
\colorlet{tablerowcolor}{gray!10} 
\title{Spectral Decomposition of Euler-Mellin Integrals
\\
 }

\author{Martin Helmer\orcid{0000-0002-9170-8295}$\,{}^1$}
\author{Felix Tellander\orcid{0000-0001-6418-8047}$\,{}^2$}
\email{\href{mailto:martin.helmer@swansea.ac.uk}{martin.helmer@swansea.ac.uk}}
\email{\href{mailto:felix@tellander.se}{felix@tellander.se}}
\address[1]{Department of Mathematics, Swansea University, Swansea, Wales, UK.}
\address[2]{Mathematical Institute, University of Oxford, Oxford OX2 6GG, UK}
\date{\today}

\begin{document}

\maketitle
\begin{abstract}
We consider the spectral decomposition of singularities of integrals and their integrands. Our results apply to any integral of Euler-Mellin type, and thus especially to every scalar Feynman integral.
Specifically we provide for both the integrand and integral respectively; two explicit constructions of the characteristic
variety and characteristic cycle of the constructible function and $D$-module they are associated with.   
From this we also obtain the singular locus or Landau singularities of the integral. En route we give a simple procedure to compute the local Euler obstruction function of a variety, and using this, to compute the Euler characteristic of the complex link of a Whitney stratum.
\end{abstract}
\section{Introduction}
Feynman integrals are a fundamental object underpinning a wide array of topics and calculations in modern physics \cite{Weinzierl:2022eaz}. As such there is significant interest in both understanding their mathematical structures and in computing them effectively. More broadly these integrals fall into the family of integrals of {\em Euler-Mellin} type, a generalization of the integrals studied in \cite{Nilsson2013,Berkesch2014}; that is integrals of the form
\begin{equation}\label{eq:EM_Int}
   \mathcal{I}(z):=\int_{\RR_+^n} \frac{x_1^{\nu_1}\cdots x_n^{\nu_n}}{f_1(x,z)^{\lambda_1}\cdots f_p(x,z)^{\lambda_p}}\frac{dx_1}{x_1}\wedge\cdots\wedge\frac{dx_n}{x_n} 
\end{equation}
for any polynomials $f_i(x,z)$ and fixed (but generic) constants $\lambda_i$ and $\nu_j$ where we treat $x$ as integration variables and $z$ as parameters.

In particular we seek to understand the singularities of the integral \eqref{eq:EM_Int}, that is the regions in $\CC^m\ni z$ where the function (or distribution) $\cI(z)$ fails to be analytic. More precisely we seek to calculate the characteristic variety and cycle of the $D$-module which annihilates $\cI(z)$ in a practical and effective manner. 

In physics integrals of the form \eqref{eq:EM_Int} with a single denominator polynomial $f(x,z)$ correspond to scalar Feynman integrals in the Lee-Pomeransky \cite{Lee2013} representation. The singularities of these integrals are called \emph{Landau singularities}~\cite{Landau:1959fi} and studied in detail in e.g.~\cite{Eden:1966dnq,Nakanishi1971}. Knowing these singularities in advance may significantly aid the evaluation of the integrals, for example with the canonical differential equations approach~\cite{Henn:2013pwa}: They may constrain or even fully predict~\cite{Dlapa:2023cvx} the analytic building blocks or (symbol) \emph{letters}~\cite{Goncharov:2010jf} of this approach. This turns the derivation of the differential equations from a symbolic, to a much simpler numerical problem. Thanks to their utility, also in other aspects of Feynman integration~\cite{Gardi:2022khw}, Landau analysis is currently experiencing a revival, see for example~\cite{Brown:2009ta,Panzer:2014caa,Klausen:2021yrt,Mizera:2021icv,Berghoff:2022mqu,Dlapa:2023cvx,Fevola:2023short,Helmer:2024wax,Caron-Huot:2024brh,Hannesdottir:2024hke} and many others.

The physicist's Landau singularity is directly related to the characteristic variety of the $D$-module annihilating \eqref{eq:EM_Int} by simply removing the zero-section and projecting onto the space of $z$-variables. In this case, Theorem \ref{thm:A} below gives us the Landau singularities and Theorem \ref{thm:C} (see also Section \ref{sec: integration}) implies that the singular locus coincides with the so called \emph{Euler discriminant} at codimension one, something that was recently proved in \cite[Theorem 5.3]{Fevola:2024acq}.

In this note we will give two recipes which can be used to obtain this characteristic variety, and from it, the characteristic cycle. The first works by computing the blow-up of (a compactification of) the ambient space of the hypersurface defined by $F=f_1\cdots f_p$ along a Jacobian ideal derived from a certain hypersurface associated to $F$. The second employs a coarsening of the Whitney stratification of the union of the (homogenization of the) hypersurface defined by $F$ with the coordinate hyperplanes $x_i=0$; the coarsening is obtained from certain multiplicity calculations for each strata.

We begin by considering the homogenization $F_h$ of the product $F=f_1\cdots f_p$ in \eqref{eq:EM_Int} as a polynomial in $x$ with respect to the new variable $x_0$; we then consider $F_h(x,z)$ as a polynomial in the ring $\CC[x_0,\dots, x_n, z_1, \dots, z_m]$ and define $G=x_0\cdots x_n F_h$. 

Consider the blow-up ${\rm Bl}_{{\rm Jac}(G)}(\pp^n\times \CC^m)$ of $\pp^n\times \CC^m$ along the ideal $${\rm Jac}(G)=\langle \frac{\partial G}{\partial x_0}, \dots, \frac{\partial G}{\partial x_n}, \frac{\partial G}{\partial z_1}, \dots ,\frac{\partial G}{\partial z_m} \rangle,$$ and let $I_{\rm Bl}$ be the ideal defining ${\rm Bl}_{{\rm Jac}(G)}(\pp^n\times \CC^m)$ in $$\CC[x,z,\xi,\zeta]=\CC[x_0,\dots, x_n, z_1, \dots, z_m, \xi_0, \dots, \xi_n, \zeta_1, \dots, \zeta_m].$$ Then, considering ${\rm Jac}(G)$ as an ideal in $\CC[x,z,\xi,\zeta]$, the exceptional divisor of the blow-up is defined by the ideal $I_E={\rm Jac}(G)+I_{\rm Bl}$.

\renewcommand*{\thetheorem}{\Alph{theorem}}
\begin{theorem}[Characteristic Variety of Integral via Blow-up]With the notations above, and letting $\scrM$ be the $D$-module which annihilates the integral \eqref{eq:EM_Int}, if we set $$I=I_{E}\cap I_{\Con(Y)}\cap I_{\Con(0)}\;{\rm and}\; I_{\Char}=(I+\langle \xi_0, \dots, \xi_n\rangle)\cap \CC[z,\zeta]$$
then the characteristic variety of $\scrM$ is given by
 $\Char(\scrM)=\bV(I_{\Char}).$ 
    \label{thm:A}
\end{theorem}

Consider a Whitney stratification $\{S_\gamma \}$ of a variety  $Y\subset \CC^n$ of dimension $d$ and suppose our chosen stratification has a total of $r$ strata with $d_\alpha=\dim (S_\alpha)$ . For a pair of strata $S_\alpha$, $S_\beta$, with ${S}_\alpha \subset \overline{S}_\beta$, let $\Eu_{\overline{S}_\beta}:\overline{S}_\beta \to \ZZ$ denote the {\em local Euler obstruction} function  associated to $S_\beta$ and let $\Eu_{\overline{S}_\beta}(S_\alpha)$ be the (constant) integer value $\Eu_{\overline{S}_\beta}(x)$ for any point $x$ in $S_\alpha$.  If $S_\alpha \not\subset \overline{S}_\beta$ then $\Eu_{\overline{S}_\beta}(S_\alpha)=0$.

Associated to the constructable function $\mathbb{1}_Y$ we have the {\em characteristic cycle}  which is the cycle class given by the formal sum $$CC(\mathbb{1}_Y) =\sum_{i=1}^r m_{\alpha_i} \Con (S_{\alpha_i}),$$where $ \Con (S_\alpha)$ is the Zariski closure of the conormal bundle to $S_\alpha$ and where the integers $m_\alpha$ are computed by solving the following explicit linear system:
\begin{equation}\hspace{-0.3cm}
    \begin{pmatrix}
        1\\
        \vdots \\
        1
    \end{pmatrix}=\Eu(Y)\begin{pmatrix}
        m_{\alpha_1}\\
        \vdots \\
        m_{\alpha_r}
    \end{pmatrix}=\begin{pmatrix}
       (-1)^{d-d_1} \Eu_{\overline{S}_{\alpha_1}}(S_{\alpha_1}) & \cdots &(-1)^{d-d_r
       }\Eu_{\overline{S}_{\alpha_r}}(S_{\alpha_1})\\
        \vdots & \ddots & \vdots\\
        (-1)^{d-d_1}\Eu_{\overline{S}_{\alpha_1}}(S_{\alpha_r}) & \cdots &(-1)^{d-d_r}\Eu_{\overline{S}_{\alpha_r}}(S_{\alpha_r})
    \end{pmatrix}\begin{pmatrix}
        m_{\alpha_1}\\
        \vdots \\
        m_{\alpha_r}
    \end{pmatrix},\label{eq:Eu_Linear_System}
\end{equation}with $d_i=d_{\alpha_i}=\dim(S_{\alpha_i})$. We emphasize that the above is, in fact, an explicit linear system with the entries of the $r\times r$ integer matrix $\Eu(Y)$ being computed from the multiplicities of subvarieties in certain {\em polar varieties} as in \eqref{eq:Eu_polar_mult}, see Remark \ref{remark:EUofComplexLink}. Note that if we arrange the strata so that $\dim(S_{\alpha_i})$ is monotonically increasing with the index $i$ the resulting integer matrix $\Eu(Y)$ is upper triangular. 

\begin{theorem}[Characteristic Variety of Integral via Whitney Stratification]\label{thm:B}Using the notation above consider $G=x_0\cdots x_n F_h$. Let $Y=\bV(G)$, take $\{S_i \;|\; i=1,\dots, r \}$ to be a Whitney stratification of $Y$, and let  $m_1, \dots, m_r$ denote the corresponding multiplicities obtained by solving the integer linear system \eqref{eq:Eu_Linear_System} using the matrix $\Eu(Y)$ resulting from this choice of Whitney stratification for $Y$. Set $\Theta\subset \{1,\dots, r\}$ to be the subset of indices for which $m_i$ is non-zero. With $I_{E}$ and $I=I_{\rm Bl}\cap I_{\Con(Y)}\cap I_{\Con(0)}$ as in Theorem \ref{thm:A} we have that
    $$\bV{(I)}=\bigcup_{i\in \Theta} {\Con(S_i)} \cup \Con(0).$$
Or in other words, the characteristic variety may be obtained from the parts of the Whitney stratification with non-zero multiplicity.  
\end{theorem}
Our final result employs either one of the two results above and gives an explicit recipe to obtain the characteristic cycle of the integral \eqref{eq:EM_Int}. As above consider  the polynomial $G=x_0\cdots x_n F_h$ in $\CC[x, z]$, where $F_h$ is the homogenization of the product $F=f_1\cdots f_p$ in \eqref{eq:EM_Int} with respect to the new variable $x_0$. Again take $Y=\bV(G) \subset \pp^n\times \CC^m$, let $\varphi: \pp^n\times \CC^m\to \CC^m $ be the projection, and let $\scrM$ denote the $D$-module which annihilates the integral \eqref{eq:EM_Int}. From either Theorem \ref{thm:A} or Theorem \ref{thm:B} we obtain the characteristic variety $\Char(\scrM)\subset \CC_z^m\times \PP^{m-1}_\zeta$ of the integral. We denote by $\Char^*(\scrM)$ the characteristic variety with the zero-section removed.  Considering the projection map $\pi_z:\CC_z^m\times \PP^{m-1}_\zeta\to \CC_z^m$ we have that the singular locus of $\scrM$ is given by $Z:={\rm Sing}(\scrM)=\overline{\pi_z(\Char^*(\scrM))}\subset \CC_z^m$. Let $\{W_i \; |i= 1, \dots, \rho\}$ be a Whitney stratification of the algebraic variety ${\rm Sing}(\scrM)$ with some fixed ordering of strata. Similar to the case of Theorem \ref{thm:B} we construct an explicit linear system arising from a $\rho \times \rho$ integer matrix of Euler obstructions of the strata. Set $c(G):=\chi((\CC^*)^n-\bV(x_1\cdots x_n f_1\cdots f_p))-(n+1)$ where the polynomials $f_i$ are  evaluated at a generic point $z_0$ outside of the singular locus $Z$, and the $x_i$ are treated as variables with the Euler characteristic calculation taking place in $\CC^n$ (since all the $x_i$ appear in the second part of the difference we can replace $(\CC^*)^n$ with $\CC^n$ or even $\PP^n$ in a straightforward way for computational purposes). We emphasize that this is an explicit and computable integer, see Remark \ref{remark:compMultCCIntegral} for further details. In particular we consider the linear system \footnotesize\begin{equation}\label{eq:CCIntMultsEu}
    \begin{pmatrix}
       c(G) +\chi(Y\cap\varphi^{-1}(z_1))\\
        \vdots\\
        c(G)+\chi(Y\cap\varphi^{-1}(z_\rho))
    \end{pmatrix}=\Eu(Z)\begin{pmatrix}
        \mu_{1}\\
        \vdots \\
        \mu_{\rho}
    \end{pmatrix}=\begin{pmatrix}
       (-1)^{d-d_1} \Eu_{\overline{W}_{1}}(W_{1}) & \cdots &(-1)^{d-d_\rho
       }\Eu_{\overline{W}_{\rho}}(W_{1})\\
        \vdots & \ddots & \vdots\\
        (-1)^{d-d_1}\Eu_{\overline{W}_{1}}(W_{\rho}) & \cdots &(-1)^{d-d_\rho}\Eu_{\overline{W}_{\rho}}(W_{\rho})
    \end{pmatrix}\begin{pmatrix}
        \mu_{1}\\
        \vdots \\
        \mu_{\rho}
    \end{pmatrix},
\end{equation}\normalsize where $z_i$ is a generic point in $W_i$. As in Theorem 
\ref{thm:B} we again emphasize that the above is, in fact, an explicit linear system that is upper triangular if the strata are ordered by dimension. The entries of the $\rho\times \rho$ integer matrix $\Eu(Z)$ are computed from the multiplicities of polar varieties as in \eqref{eq:Eu_polar_mult}, see Remark \ref{remark:EUofComplexLink}. Similarly the integers  $\chi(Y\cap\varphi^{-1}(z_i))$ may also be computed explicitly, see Remark \ref{remark:compMultCCIntegral}. 
\begin{theorem}[Characteristic Cycle of Integral via Whitney Stratification]\label{thm:C}
    Let $\scrM$ be the $D$-module which annihilates the integral \eqref{eq:EM_Int} and let $\Char(\scrM)\subset \CC_z^m\times \PP^{m-1}_\zeta$ be its characteristic variety with the associated projection map $\pi_z:\CC_z^m\times \PP^{m-1}_\zeta\to \CC_z^m$.  If $\{W_i \; |\,i= 1, \dots, \rho\}$ is a Whitney stratification of the algebraic variety $Z:={\rm Sing}(\scrM)=\overline{\pi_z(\Char^*(\scrM))}\subset \CC_z^m$ then the characteristic cycle of $\scrM$ is given by $$
    CC(\scrM)=\mu_0\Con(0)+\sum_{i=1}^\rho \mu_iW_i,
    $$with $\mu_0=c(G)+n+1$, $\mu_i$ for $i=1,\ldots,\rho$ are obtained by solving the linear system \eqref{eq:CCIntMultsEu} and $\Char(\scrM)$ is obtained via Theorem \ref{thm:A} or Theorem \ref{thm:B}.
\end{theorem}
\renewcommand*{\thetheorem}{\arabic{theorem}}
From a computational perspective, perhaps surprisingly, using the recent algorithm of \cite{helmer2023effective} to compute the Whitney stratification and the algorithm of \cite{Harris_2019} to compute multiplicities is often the most effective approach. The added benefit of this later approach is that we obtain $CC(\mathbb{1}_Y)$, as well as the characteristic cycle of the $D$-module associated to the annihilator ideal of $G$.

To emphasize the practicality of our approach we have implemented these methods in the \texttt{WhitneyStatifications} Macaulay2 \cite{M2} package which may be downloaded at the link below:
\begin{center}
\vspace{-0.2cm}
    \url{http://martin-helmer.com/Software/WhitStrat/}
\end{center}
\medskip

The paper is arranged as follows. In Section \ref{sec: Whitney stratifications and euler obstructions} we introduce the necessary geometric constructions like Whitney stratifications, Euler obstructions and polar varieties. Following this in Section \ref{sec: microlocal Dmod} we describe the annihilator ideal for one and several polynomials, their characteristic varieties, and how they relate to the microlocal description of the singularities of the solutions to these modules.  In Section \ref{sec: sheaves} we use the Riemann-Hilbert correspondence and Kashiwara's index theorem to calculate the characteristic cycle of the $D$-module annihilating the product of polynomials. This section also gives a method to compute the Euler characteristic of the complex link of a Whitney stratum (Remark \ref{remark:EUofComplexLink}). The integration of distributions and $D$-modules are described in Section \ref{sec: integration} where we also explicitly calculate the characteristic variety and cycle of several Euler-Mellin integrals.

We conclude this section with a simple example illustrating the main results. We note that in the example below we will present a Whitney stratification of a variety $Y\subset \CC^n$ in terms of closed sets to mirror the presentation used on a computer when computing such a stratification algorithmically.

\begin{example*}[Simple example of Theorem's \ref{thm:A} and \ref{thm:B}]
The Euler-Mellin integral
\small
\begin{equation*}
        \int_{\RR_+^2}\frac{x_1^{\nu_1}x_2^{\nu_2}}{(x_1+x_2+(z_1+z_2-z_3)x_1x_2+z_1x_1^2+z_2x_2^2)^{D/2}}\frac{dx_1\wedge dx_2}{x_1x_2}
    \end{equation*}
    \normalsize
    correspond to the Feynman integral for the process:\newpage
\begin{center}
\vspace*{-0.6cm}
\begin{tikzpicture}[baseline=-\the\dimexpr\fontdimen22\textfont2\relax,transform shape,scale=0.85]
    \begin{feynman}
    \vertex (a);
    \vertex [right = of a] (b);
    \vertex [right = of b] (c);
    \vertex [right = of c] (d);
    \diagram* {
	    (a) --[fermion] (b) -- [half left,edge label'=\({x_1}\)](c) -- [anti fermion](d), (c) -- [half left,edge label'=\({x_2}\)](b),
};
    \end{feynman}
    \end{tikzpicture}
\vspace{-0.2cm}
\end{center}
To use Theorem \ref{thm:A} to determine the characteristic variety of the integral we first need the characteristic variety of the modified homogenized integrand. We can obtain this directly using the blow-up construction in \ref{thm:A} or using the Whitney stratification method in Theorem \ref{thm:B}, we do the latter in this example. 

We define the polynomial of the homogenized integrand
\small
\begin{equation*}
    G=x_0x_1x_2(x_0(x_1+x_2)+(z_1+z_2-z_3)x_1x_2+z_1x_1^2+z_2x_2^2)\in\CC[x_0,\ldots,x_2,z_1,\ldots,z_3]
\end{equation*}
\normalsize
which defines a hypersurface $Y=\bV(G)\subset\PP_x^2\times\CC_z^3$. There are 24 strata in the minimal Whitney stratification of this surface and these strata consist of one of dimension zero, six of dimension one, four of dimension two, three of dimension three, six of dimension four, and are given below:
\small
\begin{align*}
  Y_0=&\bV(x_0, x_1, x_2, z_1, z_2, z_3)\\
  Y_1=&\bV(z_{3},z_{1},x_{2},x_{1},x_{0})\cup \bV(z_{3},z_{1}-z_{2},x_{2},x_{1},x_{0})\cup \bV(z_{2},z_{1},x_{2},x_{1},x_{0})\cup\bV(z_{2},z_{1}-z_{3},x_{2},x_{1},x_{0})\\&\cup \bV(z_{3},z_{2},x_{2},x_{1},x_{0} )\cup \bV(z_{2}-z_{3},z_{1},x_{2},x_{1},x_{0})\\
    Y_2=&\bV(z_3,x_2,x_1,x_0)\cup\bV(z_2,x_2,x_1,x_0)\cup\bV(z_1,x_2,x_1,x_0)\\&
    \cup\bV(x_2,x_1,x_0,z_1^2-2z_1z_2+z_2^2-2z_1z_3-2z_2z_3+z_3^2)\\
    Y_3=&\bV(x_2,x_1,x_0)\cup \bV(z_1,x_2,x_0)\cup \bV(z_2,x_1,x_0)\\
    Y_4=&\bV(x_1,x_0)\cup \bV(x_2,x_0)\cup \bV(x_1,x_2z_2+x_0),
       \bV(x_2,x_1z_1+x_0)\cup \bV(x_2,x_1)\\
       &\cup\bV(x_0,x_1^2z_1+x_1x_2z_1+x_1x_2z_2+x_2^2z_2-x_1x_2z_3)\\
    Y_5=&Y=\bV (x_0)\cup \bV(x_1^2z_1+x_1x_2z_1+x_1x_2z_2+x_2^2z_2-x_1x_2z_3+x_
      0x_1+x_0x_2)\cup \bV (x_2)\cup \bV( x_1).
\end{align*}\normalsize
Using the ordering implied above, i.e.~with the one strata of dimension zero first and the dimension five strata arising from $\bV(x_1)$ last, the matrix $\Eu(Y)$ is

\scriptsize
\begin{equation*}
\left(\!\begin{array}{cccccccccccccccccccccccc}
-1&1&1&1&1&1&1&-1&-1&-1&0&1&1&1&-1&-1&-1&-1&-1&-2&1&1&1&1\\
0&1&0&0&0&0&0&-1&0&-1&0&1&1&0&-1&-1&-1&-1&-1&-2&1&1&1&1\\
0&0&1&0&0&0&0&-1&0&0&-1&1&0&0&-1&-1&-1&-1&-1&-1&1&1&1&1\\
0&0&0&1&0&0&0&0&-1&-1&0&1&1&1&-1&-1&-1&-1&-1&-2&1&0&1&1\\
0&0&0&0&1&0&0&0&-1&0&-1&1&0&1&-1&-1&-1&-1&-1&-1&1&0&1&1\\
0&0&0&0&0&1&0&-1&-1&0&0&1&0&1&-1&-1&-1&-1&-1&-2&1&1&1&1\\
0&0&0&0&0&0&1&0&0&-1&-1&1&1&0&-1&-1&-1&-1&-1&-1&1&0&1&1\\
0&0&0&0&0&0&0&-1&0&0&0&1&0&0&-1&-1&-1&-1&-1&-2&1&1&1&1\\
0&0&0&0&0&0&0&0&-1&0&0&1&0&1&-1&-1&-1&-1&-1&-2&1&0&1&1\\
0&0&0&0&0&0&0&0&0&-1&0&1&1&0&-1&-1&-1&-1&-1&-2&1&0&1&1\\
0&0&0&0&0&0&0&0&0&0&-1&1&0&0&-1&-1&-1&-1&-1&-1&1&0&1&1\\
0&0&0&0&0&0&0&0&0&0&0&1&0&0&-1&-1&-1&-1&-1&-2&1&0&1&1\\
0&0&0&0&0&0&0&0&0&0&0&0&1&0&0&-1&0&-1&0&-1&1&1&1&0\\
0&0&0&0&0&0&0&0&0&0&0&0&0&1&-1&0&-1&0&0&-1&1&1&0&1\\
0&0&0&0&0&0&0&0&0&0&0&0&0&0&-1&0&0&0&0&0&1&0&0&1\\
0&0&0&0&0&0&0&0&0&0&0&0&0&0&0&-1&0&0&0&0&1&0&1&0\\
0&0&0&0&0&0&0&0&0&0&0&0&0&0&0&0&-1&0&0&0&0&1&0&1\\
0&0&0&0&0&0&0&0&0&0&0&0&0&0&0&0&0&-1&0&0&0&1&1&0\\
0&0&0&0&0&0&0&0&0&0&0&0&0&0&0&0&0&0&-1&0&0&1&1&1\\
0&0&0&0&0&0&0&0&0&0&0&0&0&0&0&0&0&0&0&-1&1&1&0&0\\
0&0&0&0&0&0&0&0&0&0&0&0&0&0&0&0&0&0&0&0&1&0&0&0\\
0&0&0&0&0&0&0&0&0&0&0&0&0&0&0&0&0&0&0&0&0&1&0&0\\
0&0&0&0&0&0&0&0&0&0&0&0&0&0&0&0&0&0&0&0&0&0&1&0\\
0&0&0&0&0&0&0&0&0&0&0&0&0&0&0&0&0&0&0&0&0&0&0&1
\end{array}\!\right).
\end{equation*}\normalsize
Solving the linear system \eqref{eq:Eu_Linear_System} with this matrix shows that the seven strata of smallest dimension (i.e.~all those in $Y_0$ and $Y_1$) appear with multiplicity zero in the characteristic cycle. Hence we obtain: 
 \begin{align*}
CC(\mathbb{1}_Y)&=\Con(z_3,x_2,x_1,x_0)+\Con(x_2,x_1,x_0,z_1^2-2z_1z_2+z_2^2-2z_1z_3-2z_2z_3+z_3^2)\\
    &+\Con(z_2,x_2,x_1,x_0)+\Con(z_1,x_2,x_1,x_0)+6\Con(x_2,x_1,x_0)+\Con(z_1,x_2,x_0)\\
    &+\Con(z_2,x_1,x_0)+2\Con(x_2,x_1)+\Con(x_2,x_0)+\Con(x_1,x_0)+\Con(x_2,x_1z_1+x_0)\\
    &+\Con(x_1,x_2z_2+x_0)+\Con(x_0,x_1^2z_1+(z_1+z_2-z_3)x_1x_2+x_2^2z_2)+\Con(x_0)\\
    &+\Con(x_2)+\Con(x_1)+\Con(x_1^2z_1+(z_1+z_2-z_3)x_1x_2+x_2^2z_2+x_0(x_1+x_2)).
\end{align*}
By Theorem \ref{thm:B}, the characteristic variety of the $D$-module annihilating the integrand is given by the union of the above 17 varieties together with the zero-section $\Con(0)$ and we denote the ideal defining this variety by $I$. Now using Theorem \ref{thm:A} we get that the characteristic variety of the $D$-module annihilating the integral is
\begin{align*}
    \bV(I_\Char)=&\bV((I+\avg{\xi_0,\ldots,\xi_2})\cap\CC[z,\zeta])\\
    =&\bV(\zeta_1\zeta_2+\zeta_1\zeta_3+\zeta_2\zeta_3,\, z_1\zeta_1+z_2\zeta_2+z_3\zeta_3,\, z_1\zeta_2\zeta_3-z_2\zeta_2\zeta_3-z_3\zeta_2\zeta_3-2z_3\zeta_3^2,\nonumber \\
        &{}\quad z_2\zeta_1\zeta_3-z_3\zeta_1\zeta_3+z_2\zeta_2\zeta_3-z_3\zeta_3^2)\\
        =&\Con(0)\cup\Con(z_3)\cup\Con((z_1+z_2-z_3)^2-4z_1z_2)\cup\Con(z_1)\cup\Con(z_2).\nonumber
\end{align*}
The Landau singularities of this integral can now be read of as the non-zero arguments of the $\Con(\cdot)$, i.e. these singularities are the surface defined by $z_3((z_1+z_2-z_3)^2-4z_1z_2)z_1z_2=0$.

Now that we have the characteristic variety and singular locus of the integral we can use Theorem \ref{thm:C} to calculate the characteristic variety. Let $Z=\bV(z_1z_2z_3((z_1+z_2-z_3)^2-4z_1z_2))$ be the singular locus. This variety has a Whitney stratification consisting of eleven strata:
    \small
    \begin{align*}
        Z_0&=\bV(z_1,z_2,z_3)\\
        Z_1&=\bV(z_2,z_3)\cup\bV(z_1,z_3)\cup\bV(z_1,z_2)\bV(z_1-z_2,z_3)\cup\bV(z_1-z_3,z_2)\cup\bV(z_2-z_3,z_1)\\
        Z_2&=\bV(z_3)\cup\bV(z_2)\cup\bV(z_1)\cup\bV((z_1+z_2-z_3)^2-4z_1z_2).
    \end{align*}
    \normalsize
    With the order indicated here, the Euler obstruction matrix is given by
    \small
    \begin{equation*}
        \Eu(Z)=\left(\!\begin{array}{ccccccccccc}
1&-1&-1&-1&-1&-1&-1&1&1&1&0\\
0&-1&0&0&0&0&0&1&1&0&0\\
0&0&-1&0&0&0&0&1&0&1&0\\
0&0&0&-1&0&0&0&0&1&1&0\\
0&0&0&0&-1&0&0&1&0&0&1\\
0&0&0&0&0&-1&0&0&1&0&1\\
0&0&0&0&0&0&-1&0&0&1&1\\
0&0&0&0&0&0&0&1&0&0&0\\
0&0&0&0&0&0&0&0&1&0&0\\
0&0&0&0&0&0&0&0&0&1&0\\
0&0&0&0&0&0&0&0&0&0&1
\end{array}\!\right).
    \end{equation*}
    \normalsize
    The rank of this $D$-module is three and calculating remaining necessary Euler characteristics in Theorem \ref{thm:C} we get the linear system:
    \small
    \begin{equation*}
    \begin{pmatrix}
        3\\2\\2\\2\\2\\2\\2\\1\\1\\1\\1
    \end{pmatrix}=
        \left(\!\begin{array}{ccccccccccc}
1&-1&-1&-1&-1&-1&-1&1&1&1&0\\
0&-1&0&0&0&0&0&1&1&0&0\\
0&0&-1&0&0&0&0&1&0&1&0\\
0&0&0&-1&0&0&0&0&1&1&0\\
0&0&0&0&-1&0&0&1&0&0&1\\
0&0&0&0&0&-1&0&0&1&0&1\\
0&0&0&0&0&0&-1&0&0&1&1\\
0&0&0&0&0&0&0&1&0&0&0\\
0&0&0&0&0&0&0&0&1&0&0\\
0&0&0&0&0&0&0&0&0&1&0\\
0&0&0&0&0&0&0&0&0&0&1
\end{array}\!\right)\begin{pmatrix}
        \mu_1\\\mu_2\\\mu_3\\\mu_4\\\mu_5\\\mu_6\\\mu_7\\\mu_8\\\mu_9\\\mu_{10}\\\mu_{11}
    \end{pmatrix}\iff\begin{pmatrix}
        \mu_1\\\mu_2\\\mu_3\\\mu_4\\\mu_5\\\mu_6\\\mu_7\\\mu_8\\\mu_9\\\mu_{10}\\\mu_{11}
    \end{pmatrix}=\begin{pmatrix}
        0\\0\\0\\0\\0\\0\\0\\1\\1\\1\\1
    \end{pmatrix}.
    \end{equation*}\normalsize
    This means that the characteristic cycle of this $D$-module is
    \begin{equation*}
        3\Con(0)+\Con(z_3)+\Con((z_1+z_2-z_3)^2-4z_1z_2)+\Con(z_1)+\Con(z_2).
    \end{equation*}
\end{example*}
\counterwithin{theorem}{section}

\section{Whitney Stratifications and Euler Obstructions}\label{sec: Whitney stratifications and euler obstructions}
In this section we briefly review several important geometric constructions which will play a key role in our results in later sections. 

Let $Z$ be a variety of dimension $d$. A {\em Whitney stratification} of $Z$ is a subdivision into a finite number of smooth connected manifolds $S_\alpha$, called {\em strata}, such that $Z=\bigcup_\alpha {S_\alpha}$ and such that {\em Whitney's condition} B holds for all pairs $M=S_\alpha$, $N=S_\beta$, of these manifolds. A pair of strata, $M,N$, is said to be \emph{incident}, denoted $M\preceq N$ if $M\subseteq \overline{N}$. Incident strata $M,N$ satisfy Whitney's condition B \cite[Section 19]{Whitney1965} at a point $x\in M$ if for every sequence $\{p_\ell\} \subset M$ and $\{q_\ell\}\subset N$ with $\lim p_\ell=\lim q_\ell=x$, 
  the limit of secant lines between $p_\ell,q_\ell$ is contained in the limit of tangent planes to $N$ at $q_\ell$.  We say the pair $M, N$ satisfies condition B if condition B holds at all points $x\in M$. For any $M\not\subset \overline{N}$ condition B is vacuously true. When computing Whitney stratifications, and representing them on a computer, it will often be convenient (as in e.g.~\cite{hnFOCM,helmer2023effective}) to represent the Whitney stratification as a flag $Z_\bullet$ of varieties $Z_0\subset \cdots \subset Z_d=Z$ where the strata are then the connected components of the successive differences $Z_i-Z_{i-1}$. 

  Given a variety $X=\bV(f_1, \dots, f_m)\subset \CC^n$, $\dim(X)=d$, with $f_i\in \CC[x_1,\dots, x_n]$, we let $X_{\rm Sing}$ denote the singular locus of $X$ and let  $X_{\rm reg} =X-X_{\rm Sing}$ denote the (open) manifold of smooth points in $X$. The {\em conormal variety} of a variety $X\subset \CC^n$ is the subvariety of $X
 \times \pp^{n-1}$ given by 
 \begin{align}\label{eq:conormal}
{\rm\bf Con}(X)=\overline{\left\lbrace  (p,\xi)\;|\; p \in X_{\rm reg} \text{ and } T_pX_{\rm reg} \subset\xi^\perp \right\rbrace};
\end{align} as noted in the introduction ${\rm\bf Con}(X)$ can also be thought of as the Zariski closure of the conormal bundle to the smooth part of a variety (where we hence consider the direction vectors as projective to get a unique representative, up to projective equivalence, for each direction). 
We note that we can compute equations for $\Con(X)$ in a straightforward manner. Work in the ring $\CC[x_1,\dots,x_n,\xi_1,\dots, \xi_n]$ and set $$
\mathfrak{K}=\begin{pmatrix}
\xi_1& \cdots & \xi_n\\
	\frac{\partial f_1}{\partial x_1} &\cdots& \frac{\partial f_1}{\partial x_n}\\
	\vdots& \ddots & \vdots \\
	\frac{\partial f_m}{\partial x_1} &\cdots& \frac{\partial f_m}{\partial x_n}\\
	\end{pmatrix} \quad \text{ and }
 \quad \mathfrak{J}=\begin{pmatrix}
	\frac{\partial f_1}{\partial x_1} &\cdots& \frac{\partial f_1}{\partial x_n}\\
	\vdots& \ddots & \vdots \\
	\frac{\partial f_m}{\partial x_1} &\cdots& \frac{\partial f_m}{\partial x_n}\\
	\end{pmatrix}.
$$ Let $K$ be the ideal generated by all $(n-d+1)\times (n-d+1)$ minors of the matrix $\mathfrak{K}$, and let $J$ be the ideal defined by the $(n-d)\times (n-d)$ minors of the matrix $\mathfrak{J}$. The conormal variety ${\rm\bf Con}(X)=\bV(I_{{\rm\bf Con}(X)})$ is defined by the radical ideal given by the saturation $I_{{\rm\bf Con}(X)}=(I_X+K):J^\infty$. Note also that $X_{\rm Sing}=\bV(J)$. Conormal varieties are naturally connected to Whitney stratification, see e.g.~\cite{hnFOCM}, where an algorithm to compute Whitney a stratification using conormal varieties is given. 

A related concept is that of the \emph{polar varieties} of $X$; these are defined as follows. Consider a flag $L_\bullet$ of length $d$ 
\[
L_\bullet = \left(L_{d} \supset L_{d-1} \supset \cdots \supset L_{1} \right),
\]
where each $L_i \subset \CC^n$ is an $i$-dimensional linear subspace.
  For each $i$ in $\set{0,1,\ldots,d-1}$, the dimension $i$ polar variety of $X$ along the flag $L_\bullet$ is defined as the closure 
  \[
  P_i(X;L_\bullet) := \overline{\set{p \in X_\text{reg} \mid \dim(T_pX_\text{reg} \cap L_{i+1}) > d-n+i+1}}.
  \] For $i=d$ we set $P_d(X;L_\bullet)=X$. We may write down the defining equations for the polar varieties as follows. Take $$
\mathfrak{K}_i=\begin{pmatrix}
c^{(0)}_1& \cdots & c^{(0)}_n\\
\vdots& \ddots & \vdots\\
c^{(i)}_1& \cdots & c^{(i)}_n\\
	\frac{\partial f_1}{\partial x_1} &\cdots& \frac{\partial f_1}{\partial x_n}\\
	\vdots& \ddots & \vdots \\
	\frac{\partial f_m}{\partial x_1} &\cdots& \frac{\partial f_m}{\partial x_n}\\
	\end{pmatrix},
$$ where the $c^{(i)}_j\in \CC$ are general constants for $i \in \set{0, \dots, d-1}$. 
Let $K_i$ be the ideal generated by all $(n-d+i+1)\times (n-d+i+1)$ minors of the matrix $\mathfrak{K}_i$. Then \begin{equation}
P_{i}(X;L_\bullet)=\bV((K_i+I_X):J^\infty),    \label{eq:polarVariety}
\end{equation}
where $L_{i+1}$ is the linear space spanned by the first $i+1$ rows of $\mathfrak{K}_i$. 

Polar varieties contains a lot of geometric information. For example, they can be used to effectively calculate Whitney stratifications \cite{helmer2023effective} and reduce them to a canonical coarsest stratification. Also note that the similarity in the equations of the polar and conormal varieties is not co-incidental, and in fact if we take $\pi_X:\Con(X)\to X$ to be the natural projection, the polar variety $P_{i}(X;L_\bullet)$ is the image of the (closure of the) intersection of $\Con(X)\cap L_{i+1}$ under the projection $\pi_X$ where we treat $L_{i+1}$ as an $i$ dimensional projective liner space. 

Let $Z\subset \CC^n$ be a complex variety, defined by the ideal $I_Z$ and let $W\subset Z$ be an irreducible variety defined by a prime ideal $I_W$. Recall that the {\em multiplicity} of $Z$ along
  $W$, denoted ${\rm mult}_W(Z)$, is defined as the Hilbert-Samuel multiplicity
  (see e.g.\cite[Chapter 12]{eisenbud2013commutative}) of the local ring
  $(\CC[x_1,\dots, x_n]/I_Z)_{I_W}$ at the prime ideal $I_W$; note that if $p$ is a general point in $W$ then ${\rm mult}_W(Z)={\rm mult}_p(Z)$. If $W \cap Z=\emptyset$ we define ${\rm mult}_W(Z):=0$. These multiplicities can be computed efficiently using Segre classes \cite{Harris_2019}.

  In \cite[Chapter V, Theorem 1.2]{teissier1981varietes} Teissier shows that two open strata $M,\,N$ with $M\subset\overline{N}$ satisfy Whitney's condition B if and only if the sequence of Hilbert-Samuel multiplicities
  \begin{equation*}
      \mult_\bullet(\overline{N},z)=\{\mult_zP_{0}(\overline{N};L_\bullet),\ldots,\mult_zP_{d-1}(\overline{N};L_\bullet),\,\mult_z\overline{N}\}
  \end{equation*}
is constant for every $z\in M$. Given any Whitney stratification, strata with the same sequence of polar multiplicities can be joined to obtain a unique canonical coarsest Whitney stratification.

Another geometric construction that can be obtained from the polar varieties is the \emph{local Euler obstruction}. The local Euler obstruction (or Euler obstruction) was introduced by MacPherson in his generalization of the Chern class to singular varieties \cite{MacPherson1974}, an object we now call the Chern-Schwartz-MacPherson (CSM) class. In this context, for a variety $Y$, it is thought of as a constructible function $Y\to \ZZ$ and is shown to take a constant value on each Whitney stata of $Y$. 

The same invariant was also introduced by Kashiwara a year or so earlier \cite{Kashiwara1973} and we will use his recursive definition of this invariant. Let $Y$ be a variety of dimension $d$, fix a Whitney stratification $Y=\cup S_\alpha$, with $d_\alpha=\dim(S_\alpha)$, and denote the Euler obstruction of $\overline{S}_\gamma$ along $S_\alpha$ as $\Eu_{\overline{S}_\gamma}(S_\alpha)$; as noted above the Euler obstruction takes a constant value on strata hence $\Eu_{\overline{S}_\gamma}(x)= \Eu_{\overline{S}_\gamma}(S_\alpha)$ for all $x\in S_\alpha$. We define $\Eu_{\overline{S}_\alpha}(S_\alpha)=1$ and for non-incident strata $S_\alpha$ and $\overline{S}_\beta$ we have $\Eu_{\overline{S}_\beta}(S_\alpha)=0$. Recursively we now define 
\begin{equation}\label{eq: definition Euler obstruction}
    \Eu_{\overline{S}_\gamma}(S_\alpha)=\sum_{S\alpha\prec S_\beta\preceq S_\gamma}\Eu_{\overline{S}_\gamma}(S_\beta)c(S_\alpha,\,S_\beta),\qquad c(S_\alpha,\,S_\beta)=\chi(B(x,\epsilon)\cap S_\beta\cap H_\eta)
\end{equation}
where $c(S_\alpha,\,S_\beta)$ is the Euler characteristic of a ball centered at $x\in S_\alpha$ with radius $\epsilon$ intersected with $S_\beta$ intersected with a generic linear affine space of codimension $\mathrm{dim}S_\alpha+1$ at a distance $\eta$ from $x$, where $\epsilon>>\eta>0$. This space is closely related to the \emph{complex link} of Goresky-MacPherson \cite[Part I, Sec. 2.2]{SMTbook}, the crucial difference is that the complex link is the intersection with an algebraic variety and not an open strata as above:
\begin{equation}
    CL_{\overline{S}_\beta}(S_\alpha):=B(x,\epsilon)\cap \overline{S}_\beta\cap H_\eta.
\end{equation}
Using a special case of Kashiwara's index theorem \cite{Kashiwara1973, BDK1981, Dubson1984}, which we will describe in Section \ref{sec: sheaves}, we have
\begin{equation}\label{eq: index thm for char fnc}
    \mathbb{1}_Y(S_\alpha)=\sum_{S_\alpha\preceq S_\beta}(-1)^{d-d_\alpha}m_\beta\Eu_{\overline{S}_\beta}(S_\alpha)
\end{equation}
which we can invert to
\begin{equation}
    (-1)^{d-d_\alpha}m_\alpha=1-\sum_{S_\alpha\prec S_\beta}c(S_\alpha,\,S_\beta)
\end{equation}
and using inclusion/exclusion relations for the Euler characteristic we obtain finally
\begin{equation}\label{eq:mult_complex_link}
    (-1)^{d-d_\alpha}m_\alpha=1-\chi(CL_{Y}(S_\alpha)).
\end{equation}
It is well-known in stratified Morse theory that the complex link is constant on Whitney strata. The sign convention here is chosen to guarantee that $m_\alpha\ge 0$ for perverse sheaves.

We know from the work of Teissier that the multiplicities of the polar varieties of $\overline{S}_\beta$ are constant on $S_\alpha$ \cite{teissier1981varietes}. Let $P_l(\overline{S}_\beta)$ be the polar variety of dimension $l$ with respect to a generic linear flag $L_\bullet$ . We denote the multiplicity of this polar variety at  a general point $p$ in $S_\alpha$ as ${\rm mult}_\alpha(P_l(\overline{S}_\beta)):={\rm mult}_p(P_l(\overline{S}_\beta))$. For pair of strata, $S_\alpha, S_\beta$, with  $S_\alpha\subset \overline{S_\beta}$, the polar multiplicity ${\rm mult}_\alpha(P_l(\overline{S}_\beta)=0$ for $l\le d_\alpha$. The Euler obstruction can now be calculated from the polar multiplicities \cite[Corollary 5.1.2]{trang1981varietes}
\begin{equation}\label{eq:Eu_polar_mult}
    \Eu_{\overline{S}_\beta}(S_\alpha)=\sum_{k=d_\alpha+1}^{d_\beta}(-1)^{d_\beta-k}{\rm mult}_{\alpha}(P_k(\overline{S}_\beta)).
\end{equation}
If $\cup S_\alpha$ is a Whitney stratification of the hypersurface $Y=\bV(f)$ we denote by $p_x(k)$ the multiplicity at $x$ of the \emph{relative} polar variety $P_k(f)$ which are also constant along Whitney strata so we can label the multiplicity as $p_\alpha(P_k(f))$. These numbers are also recursively defined by the normal polar multiplicities:
\begin{proposition}[{\cite[Proposition 5.1.1]{BMM1994}}]
    For $k\le d_\alpha$ we have $p_\alpha(P_k(f))=0$ and for $k>d_\alpha$
    \begin{equation}
        p_\alpha(P_k(f))=(-1)^{n-k}+\sum_{S_\alpha\prec S_\beta:k\le d_\beta}p_\beta(P_{d_\beta+1}(f))\left(\sum_{l=k}^{d_\beta}(-1)^{k-l}{\rm mult}_{\alpha}(P_l(\overline{S}_\beta))\right).
    \end{equation}
\end{proposition}
It follows by direct substitution that in the special case $k=d_\alpha +1$ we have $m_\alpha=p_\alpha(P_{d_\alpha+1}(f))$. 

\section{Microlocal Analysis of $D$-modules and Annihilator Ideals}\label{sec: microlocal Dmod}
The Weyl algebra $D$ on $\CC^n$ is given by 
\begin{equation}
    D:=R\avg{\partial_1,\ldots,\partial_n},\quad R=\CC[x_1,\ldots,x_n]
\end{equation}
where $[x_i,\,x_j]=[\partial_i,\,\partial_j]=0$ and $\partial_ix_j=x_j\partial_i+\delta_{ij}$ for all $1\le i,j\le n$. To replace $\CC^n$ with a  general manifold $X$, the ring $R$ is replaced by the sheaf of regular functions $\scrO_X$ and $D$ is replaced by $\scrD_X$, the sheaf of $\CC$-linear differential operators on $\scrO_X$. A system of partial differential equations $P_1u=0,\ldots,P_ku=0$ is encoded as the left-ideal $I=\avg{P_1,\ldots,P_k}$ and generates the $D$-module $M=D/I.$

For a non-zero element $p$ in the Weyl algebra we define the \emph{principal symbol} as the following initial form (cf.~\cite[Chapter 1]{SST}):
\begin{equation}\label{eq: principal symbol}
    \sigma(p):=\mathrm{in}_{(0,1)}(p).
\end{equation}
This is an element in the associated graded ring $\mathrm{gr}_{(0,1)}(D)=\CC[x_1,\ldots,x_n,\xi_1,\ldots,\xi_n]$ which is a commutative ring in $2n$ variables. For any left-ideal $I$ in a Weyl algebra we can similarly define the \emph{characteristic ideal} as the ideal containing all principal symbols of $p\in I$:
\begin{equation}\label{eq: characteristic ideal}
    \mathrm{in}_{(0,1)}(I):=\{\sigma(p)\,|\,p\in I\}.
\end{equation}
The \emph{characteristic variety} describes regions where the solutions to the system of PDEs (partial differential equations) associated to $I$ may be singular and is simply the variety defined by the characteristic ideal above, i.e. $$\mathrm{Char}(I):=\mathbf{V}(\mathrm{in}_{(0,1)}(I)).$$
The characteristic ideal can directly be calculated using non-commutative Gr\"obner bases in e.g. \texttt{Macaulay2} \cite{M2} and it is a theorem of \={O}aku \cite{Oaku1994} that this definition coincides with the analytic $D$-module definition using good filtrations for sheaves of modules. A $D$-ideal (or module) is said to be \emph{holonomic} if $\dim(\Char(I))=n$. All $D$-modules of interest in this paper are of an even nicer type called \emph{regular holonomic}, however, we do not give the technical definition here, see e.g. \cite[Definition 2.4.1]{SST} or \cite[Definition 5.2]{Kashiwara2003} for two different but equivalent definitions.

Let $\{\Lambda_\alpha\}$ denote the irreducible components of $\mathrm{Char}(I)$, or the irreducible components of the characteristic variety of any $D$-module $\scrM$. We the construct the algebraic cycle
\begin{equation}
    CC(\scrM)=\sum_{\alpha} m_\alpha\Lambda_\alpha
\end{equation}
called the \emph{characteristic cycle} where the multiplicities $m_\alpha$ are unique positive integers. For an intrinsic $D$-module definition of $m_\alpha$ see \cite[Proposition 1.8.2]{Bjork1993}. For holonomic modules $CC(\scrM)$ is a Lagrangian cycle.  For a variety $Y$ one may also define the characteristic cycle of the constructable function $\mathbb{1}_Y$, in this case the multiplicities $m_\alpha$ above are given by \eqref{eq:mult_complex_link}; see Section \ref{sec: sheaves} for further discussion.
 \subsection{Characteristic Varieties and Spectral Properties of Distributions}The spectral decomposition of a distribution $u\in\cD'(\RR^n)$ gives complete microlocal information on its singularity structure. This decomposition is called the \emph{wave front set} and describes not only where a distribution is singular but which components in the cotangent space are causing this behavior. We call a set $\Lambda\subset\RR^n\times\RR^n\setminus 0$ \emph{conic} if $(x,\xi)\in\Lambda\implies(x,\lambda\xi)\in\Lambda$ for all real $\lambda\neq 0$. This means that we could think of $\Lambda$ as being in $\RR^n\times\RR\PP^{n-1}$. For the construction of the analytic wave front set we follow the presentation by Treves \cite[Section 3.5]{Treves2022}.
\begin{definition} Let $v$ be a distribution with compact support in $\RR^n$. We say that $v$ is \emph{microanalytic} at a point $(x',\xi')\in\RR^n\times\RR^n\setminus 0$ if there is a conic neighborhood $\Lambda\ni(x',\xi')$ and positive real numbers $\kappa,\,c$ such that
\begin{equation}
    \left|\int_{\RR^n}e^{i\xi\cdot(x-y)-\kappa|\xi||x-y|^2}\nu(y)\Delta_\kappa(x-y,\xi)\,dy\right|\lesssim e^{-c|\xi|}
\end{equation}
where $\Delta_\kappa(x,\xi):=1+i\kappa(x\cdot \xi)/|\xi|$ holds for all $(x,\xi)\in\Lambda$ with $|\xi|$ sufficiently large.
\end{definition}
For a distribution $u\in\cD'(\Omega)$, with not necessarily compact support, we say that it is microanalytic at a point $(x',\xi')\in\Omega\times\RR^n\setminus 0$ if there is a distribution $v$ with compact support in $\Omega$ equal to $u$ in some neighborhood of $x'$ and microanalytic at $(x',\xi')$.
\begin{definition}\label{def: analytic wave front set}
    For a distribution $u\in\cD'(\Omega)$ we define the \emph{analytic wave front set} $WF_A(u)$ as the closed subset of $\Omega\times\RR^n\setminus 0$ consisting of points $(x',\xi')$ at which $u$ is \textbf{not} microanalytic.
\end{definition} 
The image of $WF_A(u)$ under the canonical projection $\pi:\Omega\times\RR^n\setminus 0\to\Omega:(x,\xi)\mapsto x$ is precisely the set where $u$ fails to be an analytic function. 

We will now see that the analytic wave front set of a distribution is closely related to the characteristic variety of the differential operators that annihilate it. Let $P$ be a differential operator of order $m$ with analytic coefficients defined in an open set $\Omega\subset\RR^n$. It has the form
\begin{equation*}
    P=P(x,\partial)=\sum_{|\alpha|\le m}a_\alpha(x)\partial^\alpha.
\end{equation*}
The principal symbol of $P$, denoted $\sigma_m$ is defined by
\begin{equation}
\sigma_m(x,\xi)=\sum_{|\alpha|=m}a_\alpha(x)\xi^\alpha.
\end{equation}
For polynomial coefficients this coincides with \eqref{eq: principal symbol}. The following result is by H\"ormander \cite[Theorem 8.6.1]{HormanderVolI}.
\begin{theorem}
    If $P$ is a differential operator of order $m$ with analytic coefficients, then
    \begin{equation}
        WF_A(u)\subset\Char^*(P)\cup WF_A(Pu),\quad u\in\cD'(\Omega),
    \end{equation}
    where $\Char^*(P)=\{(x,\,\xi)\in T^*\Omega\setminus 0\,|\, \sigma_m(x,\xi)=0\}$.
\end{theorem}
In particular if $P$ annihilates $u$, i.e. $Pu=0$, we have $WF_A(u)\subset\Char^*(P)$. This directly generalizes to the following result. 
\begin{corollary} Let $\scrM$ be the $\scrD_\Omega$-module consisting of all elements annihilating $u$, then $WF_A(u)\subset \Char(\scrM)\setminus 0\cap(\RR^n\times\RR^n\setminus 0)$.
\end{corollary}

\subsection{Annihilator Ideals}
Extending the coefficient field of the Weyl algebra from $\CC$ to polynomials $\CC[s]$ in a new commuting variable $s$, we get the algebra $D[s]:=D\otimes_\CC\CC[s]$. Denote by $\ann_{D[s]}(f^s)$ the left-ideal 
\begin{equation}
    \ann_{D[s]}(f^s):=\{Q\in D[s]\,|\,Q\bullet f^s=0\}
\end{equation}
for any $s\in \CC$ and any $x$ with $f(x)\neq 0$, this is the \emph{parametric annihilator}. It is a result by Bernstein \cite{Bernstein1972} that for a non-zero polynomial $f(x)$ there exists a non-zero polynomial $b(s)$ and differential operator $P(s)\in D[s]$ such that
\begin{equation}\label{eq: Bernstein-Sato}
    b(s)f(x)^{s}=P(s)\bullet f(x)^{s+1}.
\end{equation}
When $b(s)$ is chosen to be the monic polynomial of lowest possible degree, it is called the \emph{Bernstein-Sato} polynomial. 

The actual problem we are interested in is finding the annihilators of $f^\lambda$ with $\lambda\in\CC$ being treated as a fixed number. Replacing $s\to\lambda$ in $\ann_{D[s]}(f^s)$ give elements that annihilates $f^\lambda$, however, these elements might not fully generate $\ann_D(f^\lambda)$. The following theorem fully classifies when direct replacement provides the correct answer.
\begin{theorem}[{\cite[Theorem 5.3.13]{SST}}]\label{thm: exponent substitution}
    Let $\alpha_0$ be the smallest integer root of the Bernstein-Sato polynomial of $f$. If $\lambda\neq\alpha_0+k,\ k=1,\,2,\,\ldots$, then $\ann_D(f^\lambda)$ is given by direct replacement $s\to\lambda$ in $\ann_{D[s]}(f^s)$.
\end{theorem}

The annihilator ideal of $f^\lambda$ for the $\lambda$ not covered by the theorem can still be calculated by adding to the calculation of $\ann_{D[s]}(f^s)$ the computation of a specific syzygy module. It turns out that the characteristic variety of annihilator ideals is always very well-behaved, see \cite[Theorem 5.3.1]{SST} (c.f. \cite{KashiwaraKawai1981}).

\begin{theorem}\label{thm: annihilators are regular holonomi}
    The ideal $\ann_D(f^\lambda)$ is regular holonomic for any $\lambda\in\CC$.
\end{theorem}
\begin{remark}
    Throughout this work we will assume $\lambda_1,\ldots,\lambda_p,\nu_1,\ldots,\nu_n$ in the integral \eqref{eq:EM_Int} to be generic.
\end{remark}

\begin{example}\label{ex: Whitney umbrella}
    Let $f=x^2z-y^2$ be the polynomial defining the Whitney umbrella $Y=\bV(f)$. The parametric annihilator of this polynomial is the following left-ideal generated by six terms
        \begin{equation*}
            \mathrm{Ann}_{D[s]}(f^s)=\langle x\partial_x-2z\partial_z,\,
                y\partial_y+2z\partial_z-2s,\,
                xz\partial_y+y\partial_x,\,
                x^2\partial_y+2y\partial_z,\,
                2z^2\partial_y\partial_z+y\partial_x^2+z\partial_y \rangle.
        \end{equation*}
        and the Bernstein-Sato polynomial is given by $b(s)=(s+1)^2(s+3/2)$. By a simple Gr\"obner basis calculation we find the characteristic variety:
        \begin{align*}
            \mathrm{Char}(\mathrm{Ann}_{D}(f^\lambda))=&\mathbf{V}(y\eta + 2z\zeta, x\xi - 2z\zeta, xz\eta + y\xi, x^2\eta + 2y\zeta)\\
            =&\bV(y\eta+2z\zeta,x\xi-2z\zeta,z\eta^2-\xi^2,x\eta^2-2\xi\zeta,xz\eta+y\xi,x^2\eta+2y\zeta,x^2z-y^2)\cup\\
            &\bV(\zeta,y,x)\cup\bV(z,y,x)\cup\bV(\zeta,\eta,\xi)\\
            =&\Con(x^2z-y^2)\cup\Con(x,y)\cup\Con(x,y,z)\cup\Con(0).
        \end{align*}
        The three components except the zero-section $\Con(0)$ are shown in Figure \ref{fig: characteristic variety whitney umbrella}. By calculating the multiplicity of the characteristic ideal along the irreducible components in the characteristic variety we get the characteristic cycle:
        \begin{equation*}
            CC(\mathrm{Ann}_{D}(f^\lambda))=\Con(x^2z-y^2)+\Con(x,y)+\Con(x,y,z)+\Con(0).
        \end{equation*}
        \begin{figure}%
    \begin{subfigure}[t]{0.40\textwidth}
        \centering
         \includegraphics[width=0.8\linewidth]{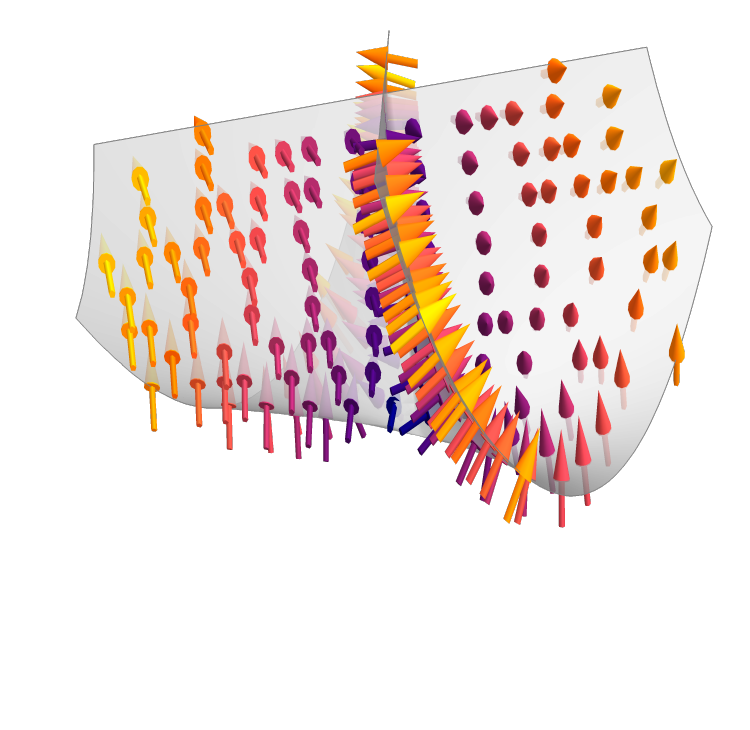}%
         \caption{ $\bV(I_{\Con(Y)})=\Con(x^2z-y^2)$}
    \end{subfigure}%
    \begin{subfigure}[t]{0.3\textwidth}
        \centering
        \vspace{-4.5cm}
         \includegraphics[width=0.8\linewidth]{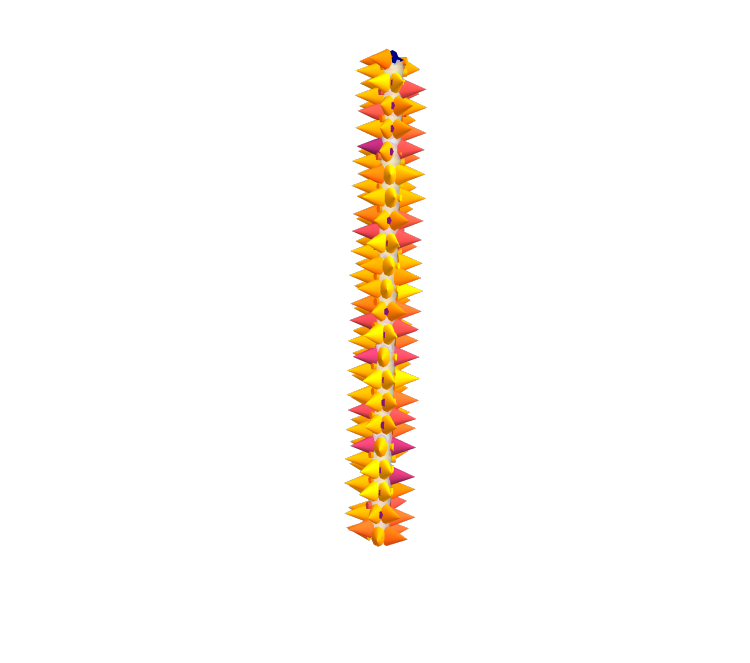}%
         \vspace{1cm}
         \caption{$\bV(\zeta,y,x)=\Con(y,x)$}
    \end{subfigure}%
    \begin{subfigure}[t]{0.3\textwidth}
        \centering
        \vspace{-4.5cm}
        \includegraphics[width=0.8\linewidth]{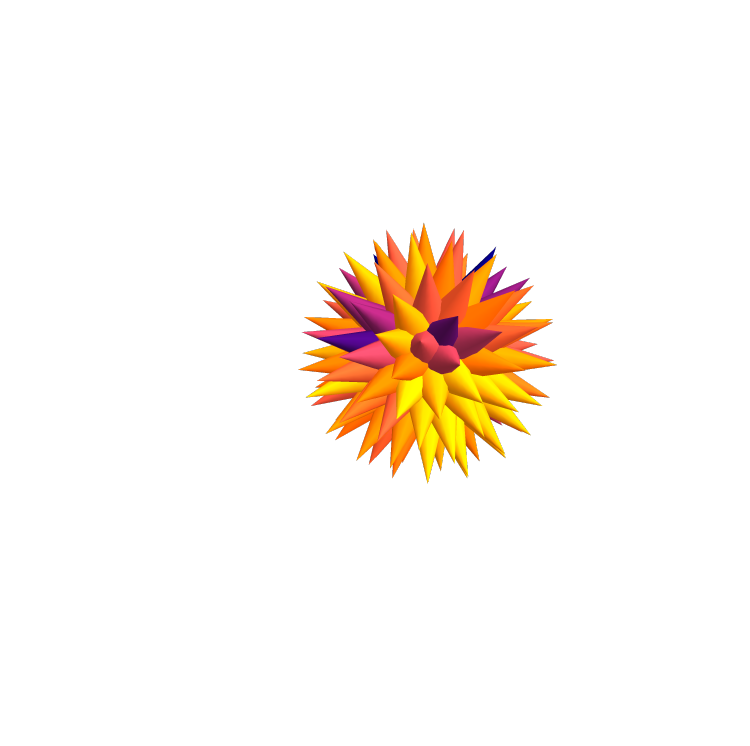}%
        \vspace{0.55cm}
        \caption{$\bV(x,y,z)=\Con(x,y,z)$}
    \end{subfigure}%
    \caption{The characteristic variety of the annihilator of $f=x^2z-y^2$; here \small $I_{\Con(Y)}:=\langle y\eta+2z\zeta,x\xi-2z\zeta,z\eta^2-\xi^2,x\eta^2-2\xi\zeta,xz\eta+y\xi,x^2\eta+2y\zeta,x^2z-y^2\rangle$. }
    \label{fig: characteristic variety whitney umbrella}
    \end{figure}
\end{example}
\subsection{The Characteristic Variety of An Annihilator Module of One Element} 
It is only very special $D$-modules where the structure of the characteristic variety is known.  If a $D$-module is invariant under a group action, the geometry of the different group orbits dictate the characteristic variety \cite[Theorem 5.1.12]{kashiwara1983systems}. An example of this appears in the theory of generalized hypergeometry by Gelfand, Kapranov and Zelevinsky (GKZ) where the group is the algebraic torus and the characteristic variety is the union of conormal bundles of the $A$-discriminants \cite{GKZ} associated to each face of the polytope of the $A$-matrix defining the system \cite[Theorem 4]{Gelfand1989}. The example from the Introduction fits within this framework which we will revisit in Example \ref{ex: bubble}, see also Example \ref{ex: GKZ integral}. We will see in this section that annihilator modules are another example where both the characteristic variety and cycle can be calculated explicitly.

\subsubsection{Kashiwara's $W_f^0$ version of Characteristic Variety}
We begin by recalling an explicit construction of the characteristic variety of $\ann_D(f^\lambda)$.

\begin{definition}\label{def:W0}
    Let $f$ be a holomorphic function on a smooth complex manifold $X$\footnote{For our purposes usually $X=\CC^n$.}.
    \begin{itemize}
        \item[(i)] ${W}_f^{s}$ is the Zariski closure of the set
        \begin{equation*}
            \{(s;\,x,\,\xi)\in\CC\times T^*X\,|\, x\in X,\ f(x)\neq 0,\ \xi= s d\log f(x)\}\quad\mathrm{in}\quad\CC\times T^*X.
        \end{equation*}
        Or in other words 
              \begin{equation*}{W}_f^{s}:=
\overline{            \left\lbrace(s;\,x,\,\xi)\in\CC\times T^*X\,|\, x\in X,\ f(x)\neq 0,\ \xi=\frac{s f'(x)}{f(x)}\right\rbrace}\quad\mathrm{in}\quad\CC\times T^*X.
        \end{equation*}
        \item[(ii)] $W_f^0:=\{(x,\,\xi)\in T^*X\,|\, (0;\, x,\,\xi)\in{W}_f^s\}$
    \end{itemize}
\end{definition}

It is a classical result that $W_f^0$, called the \emph{relative} conormal variety, is a Lagrangian variety. Examining the definition of  $W_f^0$ we see that it can be described as the vanishing set of the following ideal:
\begin{equation}
        \mathbf{I}(W_f^0)=\left(\avg{\xi_1f-\sigma(\partial_1 f),\ldots,\xi_nf-\sigma(\partial_n f)}:\avg{f}^\infty+\avg{\sigma}\right)\cap\QQ[x_1,\ldots,x_n,\xi_1,\ldots,\xi_n].\label{eq:W0Eqs}
    \end{equation}
By {\cite[Theorem 9.1]{Kashiwara2003}} we have that the variety $W_f^0$ is the characteristic variety of the ideal $\ann_D(f^\lambda)$, or in symbols:
\begin{equation}
    W_f^0=\Char(\ann_D(f^\lambda)).
\end{equation}
\begin{example} Like in Example \ref{ex: Whitney umbrella}, let 
    $f=x^2z-y^2$ be the defining equation for the Whitney umbrella. The ideal defining $W_f^0$ before eliminating $s$ is
    {\small
    \begin{align*}
        \avg{\xi(x^2z-y^2)-s(2xz),\,\eta(x^2z-y^2)-s(-2y),\,\zeta(x^2z-y^2)-s(x^2)}:\avg{x^2z-y^2}^\infty+\avg{s}=\nonumber\\
        \avg{y\eta + 2z\zeta - 2s,\, x\xi - 2z\zeta,\, xz\eta + y\xi,\, x^2\eta + 2y\zeta,\, s}
    \end{align*}}
    and eliminating $s$ gives 
    \begin{equation*}
        \mathbf{I}(W_f^0)=\avg{y\eta + 2z\zeta,\, x\xi - 2z\zeta,\, xz\eta + y\xi,\, x^2\eta + 2y\zeta}
    \end{equation*}
    Which is exactly the characteristic variety from Example \ref{ex: Whitney umbrella}.
\end{example}
\subsubsection{Computing $W_f^0$ via a Blow-up}\label{subsubsec: computing W0 via blowup} 

Fix a polynomial $f\in \CC[x_1,\dots, x_n]$ and let ${\rm Bl}_{{\rm Jac}(f)}\CC^n$ be the blow-up of $\CC^n$ along the ideal defined by $${\rm Jac}(f)=\langle \frac{\partial f}{\partial x_1}, \cdots,  \frac{\partial f}{\partial x_n} \rangle.$$ In other words, ${\rm Bl}_{{\rm Jac}(f)}X$ is the variety obtained as the closure of the graph of the rational map defined by the the generators of ${\rm Jac}(f)$:
\begin{equation}
    {\rm Bl}_{{\rm Jac}(f)}X:=\overline{\{(x, \xi)\;|\; x\in (X-\bV(\mathrm{Jac}(f))) , \;\; \xi = \mathrm{Jac}(f(x))\}} \subset X\times \PP^{n-1}.
\end{equation}
Let $E$ be the the exceptional divisor of this blow-up, that is $E={\rm Bl}_{{\rm Jac}(f)}X \cap \bV({\rm Jac}(f)) $, and take $\Lambda_\alpha$, $\alpha=2, \dots, r$, to be the irreducible components of $E$ with $C_\alpha=\pi(\Lambda_\alpha)$.

By \cite[Th\'{e}or\`{e}me 3.3]{le1983varietes}  we have the equality
\begin{equation}
    W_f^0=\Con(0)\cup\Con(\bV(f))\bigcup_{\alpha=2}^r \Lambda_\alpha.
\end{equation}Hence in particular we may compute $W_f^0$ directly from the exceptional divisor $E$ in a straightforward manner. 

\subsection{The Characteristic Variety of An Annihilator Module of a Product of Elements}\label{subsec:CharAnnProd}
The $D$-module corresponding to the integral \eqref{eq:EM_Int} is in fact associated to the annihilator ideal for a product of polynomials, rather than for the single polynomial case considered in the previous section.  It has been shown by Sabbah \cite{Sabbah1987} and Gyoja \cite{Gyoja1993} that the Bernstein-Sato construction for one polynomial can be generalized to a product of several polynomials. For $\boldsymbol{s}=(s_1,\ldots,s_p)$ we define $D[\boldsymbol{s}]:=D\otimes_\CC\CC[\ss]$ and for a product of polynomials we define the parametric annihilator ideal
\begin{equation*}
    \mathrm{Ann}_{D[\ss]}(f_1^{s_1}\cdots f_p^{s_p}):=\{P(\ss)\in D[\ss]\,|\,P(\ss)\bullet f_1^{s_1}\cdots f_p^{s_p}=0\}.
\end{equation*}
We denote by $\scrN$ the associated $D$-module $D[\ss]/\mathrm{Ann}_{D[\ss]}(f_1^{s_1}\cdots f_p^{s_p})$ and its specialization to a fixed $\boldsymbol{\lambda}=(\lambda_1,\ldots,\lambda_p)\in\CC^p$ as $\scrN_{\boldsymbol{\lambda}}:=\scrN/((s_1-\lambda_1)\scrN+\cdots+(s_p-\lambda_p)\scrN)$. As in the case with one variable, for generic $\boldsymbol{\lambda}$ we have
\begin{equation*}
    \scrN_{\boldsymbol{\lambda}}\simeq D/\ann_D(f_1^{\lambda_1}\cdots f_p^{\lambda_p}).
\end{equation*}
Using the multivariate Bernstein-Sato polynomial the genericity assumption for this isomorphism can be relaxed to an assumption of similar flavor as Theorem \ref{thm: exponent substitution}. This means that there exists an integer $N$ large enough so that fixing $\lambda_1=\cdots=\lambda_p=-N$ is generic enough and
\begin{equation*}
    D/\ann_D(f_1^{\lambda_1}\cdots f_p^{\lambda_p})\simeq D/\ann_D(F^{-N})
\end{equation*}
where $F:=f_1\cdots f_p$. Let $Y=\bV(F)$ and $\scrO_X[*Y]$ denote the sheaf of holomorphic functions with singularities along $Y$, then
\begin{equation*}
     D/\ann_D(F^{-N})\simeq \scrO_X[*Y].
\end{equation*}
 By \cite[Theorem 3.1]{le1983varietes} we get the characteristic varieties for generic $\boldsymbol{\lambda}$
\begin{equation}\label{eq: char N generic bold lambda}
\Char(\scrN_{\boldsymbol{\lambda}})=\Char(\scrO_X[*Y])=W_F^0
\end{equation}
where $W_F^0$ is defined in Definition \ref{def:W0} and can be calculated using the blow-up procedure described in the previous section.

\section{Constructible Functions and  Sheaves}\label{sec: sheaves}
We will now switch our point of view to the solutions of the system of PDEs represented by the $D$-module. These solutions will be described by sheaves for which we recommend \cite{Dimca2004sheaves, Laurentiu2022} and the classical work \cite{kashiwara-schapira1}. In the general setting we define the solution complex
\begin{equation}
    Sol(\scrM):=R\mathscr{H}om_{\scrD_X}(\scrM,\scrO_X)[n]
\end{equation}
where the elements of degree zero correspond to the classical solutions of the PDE. It is a fundamental philosophy of the theory that we can go rather freely between a $D$-module and its solution sheaf. This is formalized in the Riemann-Hilbert correspondence \cite{Kashiwara1984,Mebkhout1984,Mebkhout1984b} showing that the category of regular holonomic $D$-modules is equivalent to the category of perverse sheaves. In particular we have for holonomic $D$-modules:
\begin{theorem}[{\cite[Theorem 5.3.2]{Dimca2004sheaves}}]\label{thm: sol perserves cc}Let $\scrM$ be a holonomic $D$-module. 
With the notations above we have that    $CC(Sol(\scrM))=CC(\scrM)$ and $\Char(\scrM)=\Char(Sol(\scrM))$.
\end{theorem}

We will now define what we mean by the characteristic cycle of a complex of sheaves. Let $\scrF^\bullet$ be a constructible sheaf on the algebraic variety $X$ with respect to the Whitney stratification $\{S_\alpha\}$. For a regular map $g:X\to Y$ between algebraic varieties, let $\Phi_g(\scrF^\bullet)$ denote the \emph{vanishing cycles} as defined by Delign \cite{SGA7}. The \emph{characteristic cycle} of $\scrF^\bullet$ is the Lagrangian cycle in $T^*X$ given by
\begin{equation}
    CC(\scrF^\bullet):=\sum_\alpha m_\alpha\Con(S_\alpha).
\end{equation}
Let $x\in S_\alpha$ be a point and $g$ a generic function in a neighborhood of $x$ such that $g$ is zero on $X_\alpha$ and $(x,dg(x))\in T_{S_\alpha}^*X$ then the coefficient $m_\alpha$ is given by
\begin{equation}
    m_\alpha(\scrF^\bullet)=-\chi((\Phi_g(\scrF^\bullet))_x)
\end{equation}
i.e.~minus the Euler characteristic of the stalk of $\Phi_g(\scrF^\bullet)$ at $x$. This Euler characteristic can be calculated more easily via stratified Morse theory. For a point $p=(x,\xi)\in T_{S_\alpha}^*X$ let $NMD(\scrF^\bullet,p)$ be the \emph{normal Morse data} \cite[Definition 3.6.1]{SMTbook} and denote by $CL_X(x)$ the \emph{complex link} \cite[Part I, Sec. 2.2]{SMTbook} of $x\in S_\alpha$, then we have the equality
\begin{equation}
    -\chi((\Phi_g(\scrF^\bullet))_x)=\chi(NMD(\scrF^\bullet,p))=1-\chi(CL_X(x)),
\end{equation}hence in particular $m_\alpha(\scrF_x^\bullet)=1-\chi(CL_X(x))$.

For a holonomic $\scrD_X$-module $\scrM$ we define the \emph{index} \cite{Kashiwara1973} as
\begin{equation}\label{eq: def index}
    \chi_\scrM(x):=\sum_i(-1)^i\dim\mathscr{E}xt_{\scrD_X}^i(\scrM,\,\scrO_X)_x.
\end{equation}
Let $X_\circ\subset X$ be a dense set such that $\overline{X}_\circ=X$ and $X-X_0$ is the singular locus of $\scrM$. Then $\dim\mathscr{E}xt_{\scrD_X}^i(\scrM,\,\scrO_X)_x=0$ for $i>0$ and $x\in X_\circ$ \cite[Corollary 4.4]{kashiwara1975maximally}. The index is therefore localized at $i=0$ and since $\mathscr{E}xt^0_{\scrD_X}(\scrM,\scrO_X)=\mathscr{H}om_{\scrD_X}(\scrM,\scrO_X)$ we have
\begin{equation}
    \chi_\scrM(x)=\dim\mathscr{H}om_{\scrD_X}(\scrM,\scrO_X),
\end{equation}
i.e.~the index at a generic point is just the dimension of the solution space of the $\scrD_X$-module, sometimes referred to as the \emph{rank} of the $\scrD_X$-module. The index can also be defined for a complex of sheaves $\scrF^\bullet$ as the the Euler characteristic of the stalks: $\chi(\scrF_x^\bullet)=\sum_i(-1)^i\dim\mathscr{H}^i(\scrF_x^\bullet)$ and these numbers coincides if $\scrF^\bullet=R\mathscr{H} om_{\scrD_X}(\scrM,\scrO_X)$.

For a constructible complex of sheaves Kashiwara-Dubson's index theorem states \cite{Kashiwara1973,BDK1981,Dubson1984}
\begin{equation}
    \chi(\scrF_x^\bullet)=\sum_\alpha(-1)^{d_\alpha}m_\alpha \Eu_{\overline{S}_\alpha}(x),
\end{equation}
with $\Eu$ being the Euler obstruction from \eqref{eq: definition Euler obstruction}. 

Let $Y$ be a hypersurface defined by $F=0$ where $F=f_1\cdots f_p$ and let $\scrO_X[*Y]$ denote the sheaf of holomorphic functions with singularities along $Y$. This sheaf sits in the exact sequence
\begin{equation}
    0\to \scrO_X\to \scrO_X[*Y]\to\scrO_X[*Y]/\scrO_X\to 0
\end{equation}
and using e.g. \cite[Theorem 2.2.3]{hotta2007d} or \cite[Theorem 9.4.5]{kashiwara-schapira1} we have
\begin{equation*}
    CC(\scrO_X[*Y])=CC(\scrO_X)+CC(\scrO_X[*Y]/\scrO_X).
\end{equation*}
At the moment we view this as an equality between characteristic cycles of $D$-modules. To use Theorem \ref{thm: sol perserves cc} we first need the solution complexes. Note that, by \cite[Theorem 1.1]{Mebkhout1977}, we have 
\begin{equation*}
    R\mathscr{H} om_{\scrD_X}(\scrO_X,\,\scrO_X)=\CC_X\ \mathrm{and}\ R\mathscr{H} om_{\scrD_X}(\scrO_X[*Y]/\scrO_X,\,\scrO_X)=\CC_Y[-1].
\end{equation*} 
Shifting these by $n$ to obtain the solution complexes and using $CC(\scrF^\bullet[k])=(-1)^kCC(\scrF^\bullet)$ we find
\begin{equation}
    CC(\scrN_{\boldsymbol{\lambda}})=\Con(0)+CC(\CC_Y[d])
\end{equation}
for generic $\boldsymbol{\lambda}\in\CC^p$ and $d=\dim Y=n-1$. Applying the index theorem to $\CC_Y[d]$ and using $\chi(\CC_Y[d]_x)=(-1)^d$ we obtain
\begin{equation}\label{eq: index thm char func}
    \mathbb{1}_Y(x)=\sum_\alpha(-1)^{d-d_\alpha}m_\alpha \Eu_{\overline{S}_\alpha}(x),
\end{equation}
where $\mathbb{1}_Y$ is the characteristic function of $Y$. In particular this means
\begin{equation}
    CC(\scrN_{\boldsymbol{\lambda}})=\Con(0)+\sum_\alpha m_\alpha\Con(S_\alpha).
\end{equation}

The strata that appear in this cycle with non-zero $m_\alpha$ are precisely given by the blow-up procedure in Section \ref{subsubsec: computing W0 via blowup}. However, an alternative (and potentially faster) method that also gives us the multiplicities $m_\alpha$ explicitly is calculating a Whitney stratification of $Y$, calculating the Euler obstructions using polar multiplicities (see \eqref{eq:Eu_polar_mult}) and finally solving the linear system \eqref{eq: index thm char func}. We explain in more detail how this can be used in practice in the remark below. 
\begin{remark}[Calculating the Euler Characteristic of the Complex Link]\label{remark:EUofComplexLink}
Consider a complex variety $Y=\bV(f_1, \dots, f_s)\subset \CC^n$ of dimension $d=\dim(Y)$ and let $Y=\{S_\alpha\}$ be a Whitney stratification of $Y$. We note that the combination Kashiwara's index theorem, along with  the formulas \eqref{eq:Eu_polar_mult} and \eqref{eq:mult_complex_link}, gives us a method to compute the Euler characteristics $\chi(CL_Y(S_\alpha))$ of the complex link of any stratum $S_\alpha$ to the variety $Y$. Obviously a key step in this calculation is the computation of a Whitney stratification of $Y$, this is made possible by the recent algorithms of \cite{hnFOCM,helmer2023effective}. Starting with only the defining equations $f_1, \dots, f_s$ of the variety $Y$ the procedure is as follows:\begin{enumerate}
     \item Compute a Whitney stratification $Y=\{S_\alpha\}$ using the algorithm of \cite{hnFOCM, helmer2023effective}, e.g.~as implemented in \cite{WhitStratM2}. 
     \item Choose an ordering to place the strata in a list $S_{\alpha_1}, \dots, S_{\alpha_r}$ and, letting $d_\alpha=\dim(S_\alpha)$, construct the matrix: $$\Eu(Y)=\begin{pmatrix}
       (-1)^{d-d_1} \Eu_{\overline{S}_{\alpha_1}}(S_{\alpha_1}) & \cdots &(-1)^{d-d_r
       }\Eu_{\overline{S}_{\alpha_r}}(S_{\alpha_1})\\
        \vdots & \ddots & \vdots\\
        (-1)^{d-d_1}\Eu_{\overline{S}_{\alpha_1}}(S_{\alpha_r}) & \cdots &(-1)^{d-d_r}\Eu_{\overline{S}_{\alpha_r}}(S_{\alpha_r})
    \end{pmatrix}.$$We note that in the construction of the matrix above, if $S_\alpha \not\subset \overline{S}_\beta$ then we have $\Eu_{\overline{S}_\beta}(S_\alpha)=0$. Supposing $S_\alpha\subset \overline{S}_\beta$ we  employ  \eqref{eq:Eu_polar_mult} to compute the integer $\Eu_{\overline{S}_\beta}(S_\alpha)$;  this requires the following calculations:\begin{itemize}
        \item Compute the polar varieties  $P_k(\overline{S}_\beta)$ for $k=d_\alpha+1, \dots, d_\beta$ using \eqref{eq:polarVariety}.
        \item Compute the multiplicities $\mult_{\overline{S}_\alpha}(P_k(\overline{S}_\beta))$ using, e.g., the algorithm of \cite{Harris_2019}, as implemented in \cite{SegreM2}.
    \end{itemize}
    \item By  Kashiwara's index theorem we have the linear system: $$
       \begin{pmatrix}
        1\\
        \vdots \\
        1
    \end{pmatrix}=\Eu(Y)\begin{pmatrix}
        m_{\alpha_1}\\
        \vdots \\
        m_{\alpha_r}
    \end{pmatrix},
    $$hence we obtain the multiplicities $m_{\alpha_i}$ by solving the integer linear system above. 
    \item Using \eqref{eq:mult_complex_link} we obtain $$
    \chi(CL_Y(S_{\alpha_i}))=1-(-1)^{d-d_\alpha}m_{\alpha_i}
    $$for each strata $S_{\alpha_i}$.
 \end{enumerate}
\end{remark}

We note that it is only strata from the unique coarsest Whitney stratification that can appear with non-zero multiplicity in the characteristic cycle of $\mathbb{1}_Y$, as is proved in the proposition below.
\begin{proposition}\label{prop:unnessaryStratMult0}
    If $\{S_\alpha\}$ is a non-minimal Whitney stratification containing the unnecessary strata ${S}_\beta$ which does not appear in the unique minimal Whitney stratification, then $m_\beta=0$ in the decomposition \eqref{eq: index thm char func}.
\end{proposition}
\begin{proof}
    Let $x$ be a (smooth) point in $S_\beta$. Since \eqref{eq: index thm char func} holds for any Whitney stratification, especially the minimal one, we have
    \begin{equation}
        \mathbb{1}_Y(x)=\sum_\gamma(-1)^{d-d_\gamma}m_\gamma\Eu_{\overline{S}_\gamma}(x)=1,
    \end{equation}
    where we sum over $\overline{S}_\gamma$ lying above $\overline{S}_\beta$ in the flag.
    Including the variety $\overline{S}_\beta$ we obtain
    \begin{equation}
        1=\mathbb{1}_Y(x)=(-1)^{d-d_\beta}m_\beta\Eu_{\overline{S}_\beta}(x)+\sum_\gamma(-1)^{d-d_\gamma}m_\gamma\Eu_{\overline{S}_\gamma}(x)=(-1)^{d-d_\beta}m_\beta+1.
    \end{equation}
    Which is true only if $m_\beta=0$.
\end{proof}
Before showing some examples of calculating characteristic cycles we give some examples of just the Euler obstructions.

\begin{example}\label{ex: euler obstructions}
       \begin{figure}[t]%
    \begin{subfigure}[t]{0.33\textwidth}
        \centering
         \includegraphics[width=\linewidth]{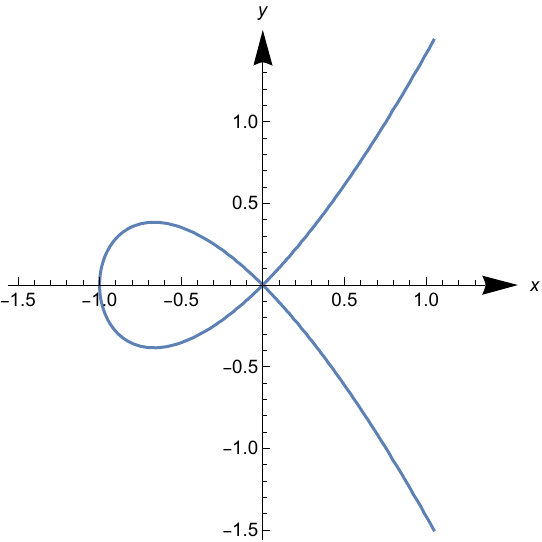}%
         \caption{$f=y^2 - (x + 1)x^2$}\label{fig: curve}
    \end{subfigure}%
    \begin{subfigure}[t]{0.33\textwidth}
        \centering
         \includegraphics[width=\linewidth]{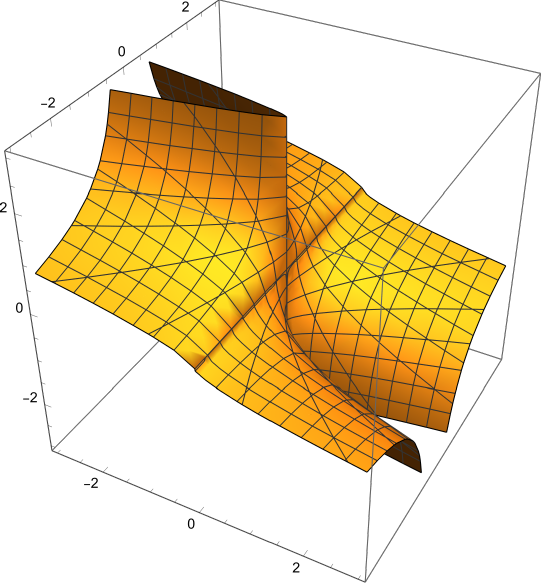}%
         \caption{$f=x^3 + y^4z^5$}\label{fig: special surface}
    \end{subfigure}%
    \begin{subfigure}[t]{0.33\textwidth}
        \centering
        \includegraphics[width=\linewidth]{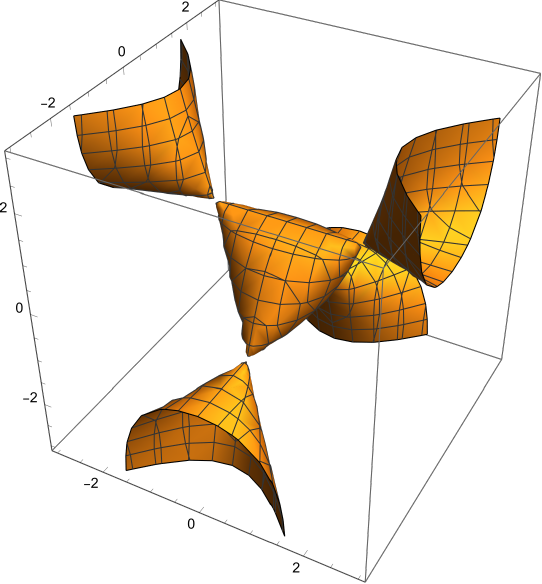}%
        \caption{\footnotesize$f=x^2 + y^2 + z^2 - x^2z + y^2z - 1$}\label{fig: cayley surface}
    \end{subfigure}%
    \caption{The three varieties used in Example \ref{ex: euler obstructions}.}
    \end{figure}

\textbf{(A)} For a curve the Euler obstructions are just the multiplicities of the points on the curve. For the curve in Figure \ref{fig: curve} every point is smooth (i.e. Euler obstruction one) except the origin where the multiplicity and therefore the Euler obstruction is two. For curves the polar multiplicity coincides with the standard multiplicity.

\textbf{(B)} For a surface $x^n+y^pz^q=0$ with $\mathrm{gcd}(n,p,q)=1$ we have $\Eu_Y(0)=\min(n,p)+\min(n,q)-n$ (see e.g. \cite{Kashiwara1973}). In the case of Figure \ref{fig: special surface} we have $\Eu_Y(0)=3$

\textbf{(C)} For a hypersurface in $\CC^n$ with isolated singularities the Euler obstructions can be obtained from the Milnor number at the singularity. If $x$ is the singular point then $\Eu_Y(x)=1+(-1)^{n}\mu$ where $\mu$ is the Milnor number associated to a generic hyperplane section of $Y$ containing $x$, see e.g.~\cite[Example 4.1.35]{Dimca2004sheaves}. For the surface in Figure \ref{fig: cayley surface} this number is one so the Euler obstruction at the singular points vanishes. Calculating the polar multiplicities at one of these points we find $\{0,2,2\}$ so indeed $\Eu_Y(x)=0$.
\end{example}

\begin{example}
    Let $Y=\bV(x^2z-y^2)$ be the affine Whitney umbrella, see Example~\ref{ex: Whitney umbrella}. Then the coarsest Whitney stratification is given by $\overline{S}_3\subset\overline{S}_2\subset\overline{S}_1=\bV(x,y,z)\subset\bV(x,y)\subset\bV(x^2z-y^2)$. The Euler obstruction is constant on open strata in a Whitney stratification and in this particular case we have the Euler obstruction matrix $\Eu(Y)$ is given in the table below:
    \begin{equation}
        \begin{array}{c|ccc}
             & \Eu_{\overline{S}_3} & \Eu_{\overline{S}_2} & \Eu_{\overline{S}_1} \\
             \hline\\
             S_3& 1 & -1 & 1 \\
             S_2&0&-1&2\\
             S_1&0&0&1
        \end{array}
    \end{equation}
    The multiplicities $m_\alpha$ are calculated by solving the linear system \eqref{eq: index thm char func}
    \begin{equation}
        \begin{pmatrix}
        1\\1\\1
        \end{pmatrix}=\begin{pmatrix}
            1&-1&1\\
            0&-1&2\\
            0&0&1
        \end{pmatrix}\begin{pmatrix}
            m_3\\m_2\\m_1
        \end{pmatrix}
    \end{equation}
    giving us that all multiplicities are 1 and therefore
    \begin{equation}
        CC(\mathbb{1}_Y)=\Con(x^2z-y^2)+\Con(x,y)+\Con(x,y,z).
    \end{equation}

    According to Proposition \ref{prop:unnessaryStratMult0} unnecessary strata always come with multiplicity zero in the characteristic cycle. We now show this for the Whitney umbrella. Add a point $\overline{S}_4=\bV(x,y,z-1)$ to the Whitney stratification such that it becomes non-minimal. The matrix of Euler obstructions then becomes
    \begin{equation}
        \begin{array}{c|cccc}
             & \Eu_{\overline{S}_4}& \Eu_{\overline{S}_3} & \Eu_{\overline{S}_2} & \Eu_{\overline{S}_1} \\
             \hline\\
             S_4 & 1& 0 & -1 & 2\\ 
             S_3&0& 1 & -1 & 1 \\
             S_2&0&0&-1&2\\
             S_1&0&0&0&1
        \end{array}
    \end{equation}
    Solving the linear system for the multiplicities, we see that the new multiplicity $m_4$ has to be zero.
\end{example}
We will now see an example showing that the converse of Proposition \ref{prop:unnessaryStratMult0} is not true, i.e. that elements in the unique coarsest Whitney stratification can appear with multiplicity zero in the characteristic cycle. The following example was used in \cite{Briancon1975} in a different context. \footnote{We thank David Massey for showing us this example.}
\begin{example}
Let $Y=\bV(z^3+ty^3z+y^4x+x^9)\subset\CC^4.$ The coarsest Whitney stratification of this variety has four strata:
    \begin{equation*}
    \{\overline{S}_4,\, \overline{S}_3\}\subset\overline{S}_2\subset\overline{S}_1=\{\bV(x,y,z,t),\,
            \bV(x,y,z,t^{12}+65536/729)\}\subset\bV(x,y,z)\subset Y.
    \end{equation*}
    Below we show the Euler obstruction matrix together with the non-trivial lists of polar multiplicities.
    \begin{equation*}
    \begin{array}{c|cccc}
             & \Eu_{\overline{S}_4}&\Eu_{\overline{S}_3} & \Eu_{\overline{S}_2} & \Eu_{\overline{S}_1} \\
             \hline\\
             S_4 & -1& 0 & 1 & -6/\{0,1,10,3\}\\
             S_3&0& -1 & 1 & -5 /\{0,1,9,3\} \\
             S_2&0&0&1&-6/\{0,0,9,3\}\\%
             S_1&0&0&0&1
        \end{array}
    \end{equation*}
    The linear system \eqref{eq: index thm char func} becomes
    \begin{equation}
        \begin{pmatrix}
        1\\1\\1\\1
        \end{pmatrix}=\begin{pmatrix}
            -1&0&1&-6\\
            0&-1&1&-5\\
            0&0&1&-6\\
            0&0&0&1
        \end{pmatrix}\begin{pmatrix}
            m_4\\m_3\\m_2\\m_1
        \end{pmatrix}\iff \begin{pmatrix}
            m_4\\m_3\\m_2\\m_1
        \end{pmatrix} =\begin{pmatrix}
            0\\1\\7\\1
        \end{pmatrix}
    \end{equation}
    so the characteristic cycle of $\mathbb{1}_Y$ is given by
    \begin{equation}
        CC(\mathbb{1}_Y)=\Con(Y)+7\Con(x,y,z)+\Con(x,y,z,t^{12}+65536/729).
    \end{equation}
\end{example}

\section{Integration}\label{sec: integration}
Before defining the integration of distributions and $D$-modules we introduce the necessary maps and spaces. 
For a morphism $\varphi:M\to N$ between manifolds we define a canonical morphism $\varphi_d$
\begin{equation}
    \varphi_d:M\times_NT^*N\to T^*M
\end{equation}
where the fibre product is defined such that
\begin{equation*}
    M\times_NT^*N=\{(m,n,\eta)\in M\times T^*N\,|\,\varphi(m)=\pi_N(n,\eta)\},
\end{equation*}where $\pi_N: T^*N\to N$ is the canonical projection. 
The morphism $\varphi_d$ is the dual of the differential of $\varphi$ which maps $T_xM\to T_{\varphi(x)}M $. Together with the canonical projection $\varphi_\pi$
\begin{equation}
    \varphi_\pi:M\times_NT^*N\to T^*N
\end{equation}
we have the commutative diagram:
\begin{equation}\label{eq: commutative f}
    \begin{tikzcd}
        T^*M \arrow[d, "\pi_M"] &\arrow[l,"\varphi_d"']  M\times_NT^*N \arrow[d, "\pi"] \arrow[r,"\varphi_\pi"] & T^*N \arrow[d,"\pi_N"] \\
         M \arrow[r,equal] &M \arrow[r,"\varphi"] & N.
\end{tikzcd}
\end{equation}
Let $X\hookrightarrow M$ and $Y\hookrightarrow N$ be algebraic varieties equipped with Whitney stratifications $\cS=\cup S_\alpha$ and $\cT=\cup T_\beta$ respectively and assume that $\varphi:X\to Y$ is a proper stratified submersion. Then the canonical projection $\varphi_\pi$
\begin{equation*}
    \varphi_\pi: (M\times_NT^*N)|_X\to T^*N|_Y
\end{equation*}
is proper and for each strata $S_\alpha\in \cS$ there exists a strata $T_\beta\in\cT$ such that
\begin{equation}
    \varphi_\pi(\varphi_d^{-1}(T_{S_\alpha}^*M))\subset T_{T_{\beta}}^*N
\end{equation}
where $T_{T_{\beta}}^*N$ is the conormal bundle. On the level of Lagrangian cycles, $\varphi_{\pi *}\varphi_d^*$ is a group homomorphism mapping Lagrangian cycles in $T^*M$ with respect to $\cS$ to Lagrangian cycles in $T^*N$ with respect to $\cT$. For a constructible function $\alpha$ on $X$ we define the push-forward
\begin{equation}\label{eq: proper push-forward constructible function}
    \varphi_*(\alpha)(y):=\int_{\{X\cap\varphi^{-1}(y)\}}\alpha d\chi
\end{equation}
for $y\in Y$. In the case $\alpha=\mathbb{1}_X$ we simply get $\varphi_*(\mathbb{1}_X)(y)=\chi(X\cap\varphi^{-1}(y))$, see e.g.~\cite[Proposition 1]{MacPherson1974} or \cite[\S4.1, Proposition 4.1.31]{Dimca2004sheaves}. The push-forward of Lagrangian cycles and constructible functions are compatible with the characteristic cycle map:
\begin{proposition}[{\cite[Proposition 3.46]{Laurentiu2022}}]
    Let $\alpha$ be a constructible function on $X$ with respect to the Whitney stratification $\cS$. Then
    \begin{equation*}
        CC(\varphi_*\alpha)=\varphi_{\pi *}\varphi_d^*CC(\alpha)
    \end{equation*}
    is a Lagrangian cycle in $T^*N$ with respect to the Whitney stratification $\cT$.
\end{proposition}
\subsection{Integration of Distributions}\label{sec: integration of distributions}
In this section we are interested in the integration of regular holonomic distributions.
\begin{definition}
    A distribution $u$ is said to be \emph{regular holonomic} if its annihilator ideal defines a regular holonomic $D$-module.
\end{definition}
As every regular holonomic distribution is tempered (cf.~\cite{Kallstrom1989}), they can always be extended to projective space $\RR_x^n\times\RR_z^m\hookrightarrow\PP_x^n\times\RR_z^m$ so that the projection $\varphi:\PP_x^n\times\RR_z^m\to\RR_z^m\,:\,(x,z)\mapsto z$ is proper. If $u(x,z)$ is a regular holonomic distribution on $\RR^n\times\RR^m$, there exists a sheaf of left ideals $\scrI$ of $\scrD_{\PP^n\times\CC^m}$ such that $\scrD_{\PP^n\times\CC^m}/\scrI$ is a sheaf of regular holonomic $\scrD_{\PP^n\times\CC^m}$-modules. If $H$ is any hyperplane and we identify $\CC^n$ with $\PP^n-H$, then $\scrI|_{\CC^n\times\CC^m}$ is equal to the sheaf of left ideals in $\scrD_{\CC^n\times\CC^m}$ annihilating $u$. Using \cite[Proposition 5]{Kashiwara1987} there is a global section $u_h$ in $\mathrm{Hom}_{\scrD_{\PP^n\times\CC^n}}(\scrD_{\PP^n\times\CC^m}/\scrI,\cD'(\PP^n\times\CC^m))$ extending $u$. This means that we identify the integration of $u$ against the measure $dx/x$ with the integration of $u_h$ against the canonical one-form $\Omega$ given by
\begin{equation*}
    \Omega=\sum_{i=0}^n(-1)^{n-1-i}dx_0/x_0\wedge\cdots\wedge\widehat{dx_i}/x_i\wedge\cdots\wedge dx_n/x_n.
\end{equation*}
This identification is well-known in the physics literature where it goes under the Cheng-Wu theorem \cite[Chapter 2, Theorem 1]{Weinzierl:2022eaz}. We therefore make the identification between the following pairings
\begin{equation*}
    \avg{\int_{\RR^n} u\frac{dx}{x},\,\phi(x,z)}:=\avg{\int_{\PP^n} u_h\,\Omega,\,\phi_h(x_0,x,z)}:=\avg{\varphi_*(u_h),\phi_h}=\avg{u_h,\,\varphi^*(\phi_h)}
\end{equation*}
where $\phi$ is a function with rapid decrease in the Schwartz space $\cS(\RR_x^n\times\RR_z^m)$, $\phi_h$ its homogenization in $x$ and $u_h$ contains a factor $x_0^{\nu_0}$ where $\nu_0$ is fixed by projective invariance.

As our objective is to make sense of the integral \eqref{eq:EM_Int} and in particular the Feynman integral in Lee-Pomeransky form:
     \begin{equation}\label{eq: Lee-Pomeransky}
         \int_{\RR_+^n}\frac{x_1^{\nu_1}\cdots x_n^{\nu_n}}{(f(x,z)-i0)^\lambda}\frac{dx_1}{x_1}\wedge\cdots\wedge\frac{dx_n}{x_n},
     \end{equation}
we have to make sure that the integrand can be understood as a regular holonomic distribution so that the above construction applies. 

We first change the integration domain from $\RR_+^n$ to $\RR^n$ by defining the the tempered distribution 
     \begin{equation*}
         x_+^\nu=x^\nu\ \mathrm{for}\ x>0,\quad x_+^\nu=0\ \mathrm{for}\ x\le 0.
     \end{equation*}
     In the Feynman integral the meaning of $i0$ is as follows. Let $f$ be a polynomial with real coefficients. For a complex parameter $\lambda$ we then define the distributions
\begin{equation}\label{eq: def i0}
    (f\pm i0)^\lambda :=f_+^\lambda +e^{\pm i\pi \lambda}f_-^\lambda
\end{equation}
where
\begin{equation*}
    f_+^\lambda=f^\lambda\ \mathrm{for}\ f>0,\quad f_+^\lambda=0\ \mathrm{for}\ f\le 0
\end{equation*}
and similarly
\begin{equation*}
    f_-^\lambda=0\ \mathrm{for}\ f\ge 0,\quad f_-^\lambda=|f|^\lambda\ \mathrm{for}\ f< 0.
\end{equation*}
For $\re\,\lambda>0$ the right-hand side in \eqref{eq: def i0} has polynomial growth and is thus well-defined as a tempered distribution. Both of the distributions $x_+^\nu$ and $f_+^\lambda$ can be extended to an entire analytic function of $\nu$ and $\lambda$ by renormalizing them with a certain product of $\Gamma$-functions coming from their Bernstein-Sato polynomial \cite[Lemma 2.10]{kashiwara1979characteristic}.

Distributions can famously not always be multiplied. A sufficient condition for multiplication to be allowed is given by Hörmander's product theorem \cite[Theorem 8.5.3] {HormanderVolI}, this condition states that the product of two distributions $u$ and $v$ can be well-defined if there exists no points $(x,\xi)\in WF_A(u),\ (y,\eta)\in WF_A(v)$ such that $x=y$ and $\xi+\eta=0$ (the colloquial term for this phenomenon is that the wave fronts \emph{collide}). The analytic wave front sets for our distributions of interest are
\begin{align}
    WF_A((x\pm i0)^\nu)&\subseteq\{(0;\,\xi)\,|\,\xi\gtrless 0\},\qquad WF_A(x_\pm^\nu)\subseteq\{(0;\xi)\,|\, \xi\in\RR\setminus 0\},
    \\
    WF_A\left(\frac{1}{(f\pm i0)^\lambda}\right)&\subseteq\bigcup_\alpha\Con(S_\alpha),
\end{align}
where $\cup S_\alpha$ is the subset of the Whitney stratification of the hypersurface $\bV(f)$ appearing with non-zero multiplicity in \eqref{eq: index thm for char fnc}. It is not hard to construct examples where the analytic wave front sets of the integrand in \eqref{eq: Lee-Pomeransky} collide. Take e.g. $f$ to be the polynomial defining the Whitney umbrella, see Example \ref{ex: Whitney umbrella}, there all three strata have colliding points with those in $WF_A(x_{i+}^\nu)$.

It is clear from Theorem \ref{thm: annihilators are regular holonomi} that all distributions in the integrand are regular holonomic. In \cite[Theorem 1]{KashiwaraKawai1979} it was claimed without proof that this product is well-defined. A more general theorem with a proof was given in \cite[Theorem 2.11]{kashiwara1979characteristic}. This theorem both tell us that (i) the product is well-defined and defines a regular holonomic distribution and (ii), that the characteristic variety of the annihilator module of the product is given by (this follows from Equation \eqref{eq: char N generic bold lambda} for generic $\lambda$)
\begin{equation}
    \Char(\mathrm{Ann}_D(f^{\lambda_1} g^{\lambda_2}))=\Char(\mathrm{Ann}_D((f\cdot g)^\lambda))
\end{equation}
for generic fixed complex numbers $\lambda_i$.

\subsection{Integration module}
We again work in the setting described at the beginning of this Section (Section \ref{sec: integration}). Let $I$ be a $D$-ideal, let $\{x_1,\ldots,x_n,z_1,\ldots,z_m\}$ local coordinates on $M$ and let $\{z_1,\ldots,z_m\}$ local coordinates on $N$. Define the left $D_N$ and right $D_M$ bimodule $D_{N\leftarrow M}:$
\begin{equation*}
    D_{N\leftarrow M}:=D_M/(\partial_1D_M+\cdots+\partial_nD_M).
\end{equation*}
We then define the \emph{integration module} as the tensor product
\begin{equation}
    D_{N\leftarrow M}\otimes_{D_M}D_M/I\simeq D_M/(I+\partial_1D_M+\cdots+\partial_nD_M)
\end{equation}
with the \emph{integration ideal} being
\begin{equation}
    (I+\partial_1D_M+\cdots+\partial_nD_M)\cap D_N.
\end{equation}

We note that the integral of a regular holonomic $D$-module is again a regular holonomic $D$-module, \cite[Theorem 5.5.3]{SST}. Given a morphism of manifolds $\varphi:M\to N$ we will denote the integral of a $D$-module $\scrM$ on $M$ along the fibres of $\varphi$ by $\int_\varphi \scrM$. In \cite{Kashiwara2003} Kashiwara proves the following theorem (see also \cite[Remark 2.5.2]{hotta2007d}). 

\begin{theorem}[Kashiwara, \cite{Kashiwara2003}, Theorem 4.27]
    Let $\varphi:M\to N$ be a morphism of manifolds and $\scrM$ a  regular holonomic $D_M$-module such that $\varphi$ is proper on $\mathrm{Supp}\,\scrM\to N$. Then
    \begin{equation}
        \Char\left(\int_\varphi \scrM \right)\subset \varphi_\pi \varphi_d^{-1}\Char(\scrM).
    \end{equation}\label{thm:KashiwaraIntegral}
\end{theorem}

\begin{remark}\label{remark:CC_integral_Finite_equal}
    From \cite[page~79]{Kashiwara2003}  we have 
    that equality holds in the above theorem if
    \begin{equation*}
        \varphi_\pi:\varphi_d^{-1}\Char(\scrM)\to T^*N
    \end{equation*}
    is finite. We will establish that this is fact always the case in our setting in Lemma \ref{lemma:int_Dmod} below.  
\end{remark}

In the special case of Euler-Mellin integrals, we can identify $M=\PP^n_x\times\CC^m_z$ and $N=\CC_z^m$ with $\varphi$ being the canonical projection 
The commutative diagram \eqref{eq: commutative f} now becomes
\begin{equation}
    \begin{tikzcd}
        \PP^n_x\times\CC^m_z\times\check{\PP}_\xi^n\times\check{\PP}_\zeta^{m-1} \arrow[d, "\pi_M"] &\arrow[l,"\varphi_d"']  \PP^n_x\times\CC^m_z\times\check{\PP}_\zeta^{m-1} \arrow[d, "\pi"] \arrow[r,"\varphi_\pi"] & \CC^m_z\times\check{\PP}_\zeta^{m-1} \arrow[d,"\pi_N"] \\
         \PP^n_x\times\CC^m_z \arrow[r,equal] &\PP^n_x\times\CC^m_z \arrow[r,"\varphi"] &  \CC^m_z.
\end{tikzcd}\label{eq:integeralCommDiagram}
\end{equation}

For the integral \eqref{eq:EM_Int} we first define the annihilator module of the integrand
\begin{equation*}    \scrN_{\boldsymbol{\lambda},\boldsymbol{\nu}}=D_{\CC^n\times\CC^m}/\ann_{D_{\CC^n\times\CC^m}}(x_1^{\nu_1}\cdots x_n^{\nu_n}f_1^{-\lambda_1}\cdots f_p^{-\lambda_p})
\end{equation*}
and then by the construction in the beginning of Section \ref{sec: integration of distributions} we lift this to a module on $\PP^n\times\CC^m$ generated by the annihilator of the product of homogenization of all the $f_i$ and $x_0^{\nu_0}\cdots x_1^{\nu_1}$.
 
In practice this means to calculate the characteristic variety of the integration ideal we simply have to calculate the fibre of $\varphi_d^{-1}$ and eliminate the integration variables and their corresponding cotangent variables. Explicitly we have that \small
\begin{equation}\label{eq:f_d_inv}\mathbf{I}(\varphi_d^{-1}\Char(\scrM))= 
    (\mathbf{I}(\Char(\scrM))+\avg{\xi_0,\ldots,\xi_n})\cap\QQ[x_0, \dots, x_n, z_1,\ldots,z_m,\zeta_1,\ldots,\zeta_m]
\end{equation}\normalsize
hence to obtain $ \varphi_\pi \varphi_d^{-1}\Char(\scrM)$ we just have to calculate (the radical of) the ideal $$(\mathbf{I}(\Char(\scrM))+\avg{\xi_0,\ldots,\xi_n})\cap\QQ[z_1,\ldots,z_m,\zeta_1,\ldots,\zeta_m].$$
We now show that the assertion of Remark \ref{remark:CC_integral_Finite_equal} holds in our case of interest. 
\begin{lemma}\label{lemma:int_Dmod}
     Let $\varphi:\PP^n_x\times \CC_z^m \to \CC_z^m$ be the canonical projection and let $\scrM$ be a regular holonomic $D_M$-module. Then
    \begin{equation}
        \Char(\int_\varphi\scrM)= \varphi_\pi \varphi_d^{-1}\Char(\scrM).
        \end{equation}
\end{lemma}
\begin{proof} Note that by Remark \ref{remark:CC_integral_Finite_equal} it is enough to show that $\varphi_\pi:\varphi_d^{-1}\Char(\scrM) \to \CC_z^m \times \check{\PP}_\zeta^{m-1}$ is a finite map; recall a map is finite if and only if it is proper and quasi-finite, i.e.~proper and has finite fibers (see e.g.~\cite[Corollary 6.5]{MumfordOda}). Note $\varphi_\pi$ is proper as all fibers are projective varieties and hence compact. Hence it is enough to show that for each $p\in \varphi_d^{-1}\Char(\scrM) $ the fiber $\varphi_\pi^{-1}(\varphi_\pi(p))$ is finite; i.e.~that the map $\varphi_\pi:\varphi_d^{-1}\Char(\scrM) \to {\varphi_\pi(\varphi_d^{-1}\Char(\scrM))}$ is a finite map. 
Now, since we only consider the map  $\varphi_\pi$ onto its image it is dominant by construction, and by, e.g.~\cite[page 323]{stasica2002effective}, we know that a dominant map between complex varieties is finite if and only if it is proper, hence the conclusion follows.
\end{proof}

\begin{remark}[Proof of Theorem \ref{thm:A} and Theorem \ref{thm:B}]
We note that the portion of the conclusions of Theorem \ref{thm:A} and Theorem \ref{thm:B} relating the characteristic variety of the integral \eqref{eq:EM_Int} to that of the integrand follow directly from Lemma \ref{lemma:int_Dmod} above. The (implicit) statement in Theorem \ref{thm:A} that the characteristic variety associated to the integrand is given by the blowup construction follows from  the discussions in Subsections \ref{subsubsec: computing W0 via blowup} and \ref{subsec:CharAnnProd} (in particular \eqref{eq: char N generic bold lambda}). The (implicit) statement in Theorem \ref{thm:B} that the characteristic variety associated to the integrand is given by the computation of conormal varieties of Whitney strata and multiplicities follows from \eqref{eq: char N generic bold lambda}, and the discussions in Section \ref{sec: sheaves}, see in particular Remark \ref{remark:EUofComplexLink} for the multiplicity portion.
\end{remark}

We will now show how to obtain not just the characteristic variety but in fact the full characteristic cycle of the $D$-module annihilating the integral. As seen in Section \ref{sec: sheaves} there is a close relationship between constructible functions and Lagrangian cycles, i.e. our characteristic cycles, this is actually an equivalence of categories proven by MacPherson~\cite{MacPherson1974}. 

Let $Y\subset M$ be a hypersurface defined defined by $f=0$. In Section \ref{sec: sheaves} we showed that the characteristic cycle of the annihilator ideal of $f$ is determined by $Y$, it follows directly that
\begin{equation}
    CC(\mathrm{Ann}_D(f^\lambda))=CC(\mathbb{1}_M-\mathbb{1}_Y).
\end{equation}
With $M=\PP_x^n\times\CC_z^m$, $\varphi$ being the canonical projection onto $\CC_z^m$ and $Y=\bV(x_0\cdots x_nF_1\cdots F_p)$ we apply the proper push-forward from \eqref{eq: proper push-forward constructible function} we obtain
\begin{equation*}
    \varphi_*(\mathbb{1}_{\PP^n\times\CC^m}-\mathbb{1}_Y)(z)=\chi((\PP^n\times\CC^m)\cap\varphi^{-1}(z))-\chi(Y\cap\varphi^{-1}(z))=n+1-\chi(Y\cap\varphi^{-1}(z)).
\end{equation*}
This constructible function is simply the index \eqref{eq: def index} of the integration module $\mathscr{J}=\int_\varphi\scrN_{\boldsymbol{\lambda},\boldsymbol{\nu}}$. For a generic $z$ (especially outside the singular locus) this gives the rank, i.e. the dimension of the solution space, of this module.

Note that the integer $\mathrm{rank}(\scrJ)$ is precisely the maximum likelihood degree \cite{Huh2013} of $\bV(x_1\cdots x_n f_1\cdots f_p)$, i.e. the Euler characteristic of its complement in the algebraic torus:
\begin{equation}
  \mathrm{rank}(\scrJ)=  \chi((\CC^*)^n-\bV(x_1\cdots x_n f_1\cdots f_p)).
\end{equation}
Let $Z={\rm Sing}(\mathscr{J})\subset \CC_z^m$ be the singular locus of  $\mathscr{J}$. If $\pi_z:\CC_z^m\times\PP_\zeta^{m-1}\to\CC_z^m$ is the canonical projection, then $\mathrm{Sing}(\mathscr{J}):=\overline{\pi_z(\Char^*(\mathscr{J}))}$. Let $\{W_\alpha\}$ be a Whitney stratification of $Z$. To get the characteristic cycle we can again employ Kashiwara's index theorem:
\begin{equation}\label{eq: index theorem integral}
    \chi_{\scrJ}(z)=n+1-\chi(Y\cap\varphi^{-1}(z))=\mathrm{rank}(\scrJ)-\sum_\alpha(-1)^{d-d_\alpha}\mu_\alpha \Eu_{\overline{W}_\alpha}(z),
\end{equation}
 where $\mu_\alpha$ is an integer (giving the multiplicity), $d$ is the dimension of the singular locus $Z$, and $d_\alpha=\dim(W_\alpha).$ Since the Euler obstruction is constant on strata, the index $\chi_\scrJ(z)$ is also constant on strata. By ordering the strata by dimension this again leads to an upper-triangular system where we can solve for the $\mu_\alpha$ explicitly. We give details as to how the values appearing in this computation may be carried our explicitly in below. 

 Finally, we note that \eqref{eq: index theorem integral} also implies \cite[Theorem 5.1]{Fevola:2024acq}.

 \begin{remark}[Computing multiplicities for the characteristic cycle of the integral]\label{remark:compMultCCIntegral} As above consider  the polynomial $G=x_0\cdots x_n F_h$ in $\CC[x,z]$, where $F_h$ is the homogenization of the product $F=f_1\cdots f_p$ in \eqref{eq:EM_Int} with respect to the new variable $x_0$. Again take $Y=\bV(G) \subset \pp_x^n\times \CC_z^m$, let $\varphi: \pp_x^n\times \CC_z^m\to \CC_z^m $ be the projection, and let $\scrM$ denote the $D$-module which annihilates the integral \eqref{eq:EM_Int}. 

From either Theorem \ref{thm:A} or Theorem \ref{thm:B} we obtain the characteristic variety $\Char(\scrM)\subset \CC_z^m\times \PP^{m-1}_\zeta$ of the integral.  Considering the projection map $\pi_z:\CC_z^m\times \PP^{m-1}_\zeta\to \CC_z^m$ we have that the singular locus of $\scrM$ is given by $Z:={\rm Sing}(\scrM)=\overline{\pi_z(\Char^*(\scrM))}\subset \CC_z^m$ where $\Char^*(\scrM)$ denotes the characteristic variety with the zero-section removed. Let $\{W_i \; |i= 1, \dots, \rho\}$, $Z=\cup W_i$, be a Whitney stratification of the algebraic variety ${\rm Sing}(\scrM)$ with some fixed ordering of strata.  Set $c(G):=\chi((\CC^*)^n-\bV(x_1\cdots x_n f_1(x,z_0)\cdots f_p(x,z_0)))-(n+1)$ where the polynomials $f_i$ are  evaluated at a generic point $z_0$ outside of the singular locus $Z$, and the $x_i$ are treated as variables with the Euler characteristic calculation taking place in $\CC^n$. 

Consider the problem of computing the multiplicities $\mu_\alpha$ in \eqref{eq: index theorem integral}. From \eqref{eq: index theorem integral} we can write down the linear system (given also above in \eqref{eq:CCIntMultsEu}): 
 \footnotesize\begin{equation*}
    \begin{pmatrix}
       c(G) +\chi(Y\cap\varphi^{-1}(z_1))\\
        \vdots\\
        c(G)+\chi(Y\cap\varphi^{-1}(z_\rho))
    \end{pmatrix}=\Eu(Z)\begin{pmatrix}
        \mu_{1}\\
        \vdots \\
        \mu_{\rho}
    \end{pmatrix}=\begin{pmatrix}
       (-1)^{d-d_1} \Eu_{\overline{W}_{1}}(W_{1}) & \cdots &(-1)^{d-d_\rho
       }\Eu_{\overline{W}_{\rho}}(W_{1})\\
        \vdots & \ddots & \vdots\\
        (-1)^{d-d_1}\Eu_{\overline{W}_{1}}(W_{\rho}) & \cdots &(-1)^{d-d_\rho}\Eu_{\overline{W}_{\rho}}(W_{\rho})
    \end{pmatrix}\begin{pmatrix}
        \mu_{1}\\
        \vdots \\
        \mu_{\rho}
    \end{pmatrix}\end{equation*}\normalsize
 Using the procedure described in Remark \ref{remark:EUofComplexLink} we can compute the integer matrix $\Eu(Z)$. It remains to compute the integers $c(G)$ and $\chi(Y\cap\varphi^{-1}(z_i))$. We first note that the Euler characteristics of affine or projective complex varieties (and of Zariski open subsets of affine or projective varieties via the inclusion/exclusion property of Euler characteristics) may be computed with any algorithm to compute the Euler characteristic of a projective variety (e.g.~\cite{helmer2016proj}, as implemented in \cite{charClassM2}). Hence the main practical question remaining is how to deal with picking random $z_i$ of the required type which can be effectively represented on a computer (i.e.~most practical sampling methods deliver floating point numbers as output which can then not be used as input into the Euler characteristic calculation in a straightforward way). 

 First consider the integer $\chi((\CC^*)^n-\bV(x_1\cdots x_n f_1(x,z_0)\cdots f_p(x,z_0)))$ appearing in $c(G)$; in this case, since we seek a general element $z_0$ of $\CC^m$ not in $Z$ we may simply choose a random vector of rational numbers (since $Z$ is a proper subvariety of $\CC^m$). Next consider the integers $\chi(Y\cap\varphi^{-1}(z_i))$, where $z_i$ is a general element in $W_i$. Let $I_{W_i}$ denote the radical ideal defining $\overline{W_i}.$ One remedy for this is as follows:\begin{enumerate}
     \item Take $d_i=\dim(W_i)$ general linear equations $\ell_j$ defining the idea $L_i=\langle \ell_{1}, \dots , \ell_{d_i} \rangle$ in $\CC[z]$, so that $\bV(L_i)$ is a general linear space of codimension $d_i$ in $\CC^m_z$. In practice we can construct $\ell_j$ by taking random rational coefficients, and hence can represent $L_i$ on a computer.  
     \item Compute the Euler characteristic of $Y \cap \bV(I_{W_i}+L_i)$ where we abuse notation and consider the ideals $I_{W_i}$ and $L_i$ both as ideals in $\CC[z]$ and in $\CC[x,z]$. Note that by construction $\bV(I_{W_i}+L_i)\subset \CC^m_z$, considered as an variety in $\CC^m_z$, consists of exactly $\deg(\overline{W_i})$ many points; let $\Theta_i$ denote this finite set of $\deg(\overline{W_i})$ many general points in $W_i$. It follows that $Y \cap \bV(I_{W_i}+L_i)=Y\cap \varphi^{-1}(\Theta_i)$. 
     \item Hence we have $$\chi(Y\cap\varphi^{-1}(z_i))=\frac{\chi(Y \cap \bV(I_{W_i}+L_i))}{\deg(\overline{W_i})}.$$
 \end{enumerate}
 \end{remark}

\begin{remark}[Proof of Theorem \ref{thm:C}]
    The proof of Theorem  \ref{thm:C} follows immediately from the discussions above in this subsection. 
\end{remark}

\begin{example}\label{ex: simple integral}

    We consider the integral
    \begin{equation*}
        \int_\sigma\frac{dx}{z-x^2}.
    \end{equation*}
    The ideal annihilating the integrand is $I=\avg{2x\partial_z+\partial_x,\, x\partial_x+2z\partial_z+2}$ and the integration ideal $J=\avg{2z\partial_z+1}.$ The characteristic varieties are simply
    \begin{align*}
        \Char(I)&=\bV(2x\zeta+\xi,\,2z\zeta+x\xi,\,x^2-z)\cup\bV(\xi,\zeta)=\Con(x^2-z)\cup\Con(0)\\
        \Char(J)&=\bV(z)\cup\bV(\zeta)=\Con(z)\cup\Con(0)
    \end{align*}
    and it is readily verified that $\Char(J)=\bV((\mathrm{in}_{(0,1)}(I)+\avg{\xi})\cap\QQ[z,\zeta])$. The student of distribution theory will recognize this example. We have the distributional identity
\begin{equation}
    \int\frac{dx}{z-x^2 + i0}=-i\pi(z+i0)^{-1/2}
\end{equation}
with the analytic wave front sets
\begin{align*}
    WF_A((z-x^2+i0)^{-1})&=\{(z,x;\,\lambda z,-2\lambda x)\in \RR^2\times\RR^2\setminus 0\,|\,z-x^2=0,\,\lambda>0\},\\ WF_A((z+i0)^{-1/2})&=\{(0;\,\zeta)\,|\,\zeta>0\},
\end{align*}
depicted in Figure \ref{fig: integral parabola normal}.
\begin{figure}[t]
    \centering
    \includegraphics[width=0.3\linewidth]{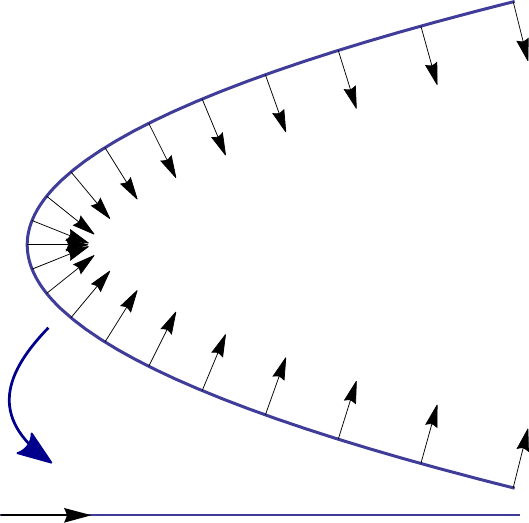}
    \caption{The analytic wave front sets in Example \ref{ex: simple integral}.}
    \label{fig: integral parabola normal}
\end{figure}
\end{example}

\begin{example}\label{ex: GKZ integral}
    We now consider one of the simplest Euler-Mellin integrals:
    \begin{equation}
        \int_0^{\infty}(z_1+z_2x_1+z_3x_1^2)^ax_1^{-b}\frac{dx_1}{x_1}.
    \end{equation}
    This integral fits within the framework of generalized hypergometry with Gelfand, Kapranov, and Zelevinsky and is annihilated hypergeometric system $H_A(\beta)$ defined by
    \begin{equation}
        A=\begin{pmatrix}
            1&1&1\\
            0&1&2
        \end{pmatrix},\qquad \beta=\begin{pmatrix}
            a\\ b
        \end{pmatrix}.
        \end{equation}
        For GKZ systems we can not only calculate the characteristic variety but in fact the full characteristic cycle directly as the sum of the conormal varieties \cite[Theorem 4 and 5]{Gelfand1989} of the factors in the principal $A$-determinant \cite[Chapter 10, Theorem 1.2]{GKZ} weighted by their exponents:
        \begin{equation}
           CC(H_\beta(A))=2\Con(0)+\Con(z_2^2-4z_1z_3)+\Con(z_1)+\Con(z_3).
        \end{equation}
        We will now calculate the characteristic variety, i.e. the union of the above terms, using the methods in this paper.
        
        Let $Y\subset\PP_x^1\times\CC_z^3$ be the hypersurface defined by $G=x_0x_1(z_1x_0^2+z_2x_0x_1+z_3x_1^2)$ equal to zero. The minimal Whitney stratification of this variety contains 13 elements. Solving the linear system \eqref{eq: index thm char func} with the Euler obstruction matrix
        \begin{equation*}
            \Eu(Y)=\left(\!\begin{array}{ccccccccccccc}
1&-1&-1&-1&0&1&1&-1&-1&-1&1&1&2\\
0&-1&0&0&1&1&0&-1&-1&0&1&1&1\\
0&0&-1&0&1&0&1&-1&0&-1&1&1&1\\
0&0&0&-1&0&1&1&-1&-1&-1&1&1&2\\
0&0&0&0&1&0&0&-1&0&0&1&1&1\\
0&0&0&0&0&1&0&-1&-1&0&1&1&2\\
0&0&0&0&0&0&1&-1&0&-1&1&1&2\\
0&0&0&0&0&0&0&-1&0&0&1&1&2\\
0&0&0&0&0&0&0&0&-1&0&1&0&1\\
0&0&0&0&0&0&0&0&0&-1&0&1&1\\
0&0&0&0&0&0&0&0&0&0&1&0&0\\
0&0&0&0&0&0&0&0&0&0&0&1&0\\
0&0&0&0&0&0&0&0&0&0&0&0&1
\end{array}\!\right)
        \end{equation*}
        shows that the four components of dimension zero and one: $\bV(x_0,x_1,z_1,z_2),\ \bV(x_0,x_1,z_1,z_3),$ $\bV(x_0,x_1,z_2,z_3),\ \bV(x_0,x_1,z_1,z_2,z_3)$, have multiplicity zero. The characteristic cycle is given by
        \begin{align*}     CC(\mathbb{1}_Y)=&\Con(z_1x_0^2+z_2x_0x_1+z_3x_1^2)+\Con(x_0)+\Con(x_1)+\Con(x_1,z_1)+\Con(x_0,z_3)\\
            & +3\Con(x_0,x_1)+\Con(x_0,x_1,z_2^2-4z_1z_3)+\Con(x_0,x_1,z_1)+\Con(x_0,x_1,z_3).
       \end{align*}

       This example is small enough so the characteristic variety of $\mathbb{1}_Y$ can also be calculated via the blow-up method. The ideal defining $\mathrm{Bl}_{\mathrm{Jac}(G)}(\PP_x^1\times\CC_z^3)$ is
       \begin{equation*}
           \avg{\zeta_2^2-\zeta_1\zeta_3,\,x_1\zeta_2-x_0\zeta_3,\,x_1\zeta_1-x_0\zeta_2,\,x_1\xi_1-z_1\zeta_1-2z_2\zeta_2-3z_3\zeta_3,\,x_0\xi_0-3z_1\zeta_1-2z_2\zeta_2-z_3\zeta_3}
       \end{equation*}
       Let $D_\alpha$ be the irreducible components of the exceptional divisor ${\rm Bl}_{\mathrm{Jac}(G)}(\PP_x^1\times\CC_z^3)\cap\mathrm{Jac}(G)$, the characteristic variety is
        \begin{align*}
            \Char(\mathbb{1}_Y)&=\bigcup D_i \cup \Con(Y)\\&=\Con(z_1x_0^2+z_2x_0x_1+z_3x_1^2)\cup\Con(x_0)\cup\Con(x_1)\cup\Con(x_1,z_1)\cup\Con(x_0,z_3)\\
            &\ \ \ \ \ \cup\Con(x_0,x_1)\cup\Con(x_0,x_1,z_2^2-4z_1z_3)\cup\Con(x_0,x_1,z_1)\cup\Con(x_0,x_1,z_3).
        \end{align*}
        
    Denote by $I$ the ideal defining the variety $\Char(\mathbb{1}_Y)$, or equivalently the support of $CC(\mathbb{1}_Y)$. The ideal defining the characteristic variety of the $D$-module is then given by $I_{\Char}=(I+\avg{\xi_0,\,\xi_1})\cap\QQ[z,\zeta]$ so the characteristic variety is given by
    \begin{align*}
        \bV(I_\Char)=&\bV(\zeta_2^2-\zeta_1\zeta_3,\,z_2\zeta_2+2z_3\zeta_3,\,z_1\zeta_1-z_3\zeta_3)\\
        =&\Con(0)\cup\Con(z_1)\cup\Con(z_3)\cup\Con(z_2^2-4z_1z_3).
    \end{align*}
    
    We now want to ccalculate the multiplicities of the components in the characteristic cycle. The singular locus $Z$ of this $D$-module is given by $\bV(E_A=z_1z_3(z_2^2-4z_1z_3))$ and the Whitney stratification of this variety contains 7 strata: 
    \small
    \begin{align*}
        Z_0&=\bV(z_1,z_2,z_3)\\
        Z_1&=\bV(z_1,z_3)\cup\bV(z_2,z_3)\cup\bV(z_1,z_2)\\
        Z_2&=\bV(z_3)\cup\bV(z_1)\cup\bV(z_2^2-4z_1z_3).
    \end{align*}
    \normalsize
    With the order indicated here, the Euler obstruction matrix is given by
\begin{equation*}
    \Eu(Z)=\left(\!\begin{array}{ccccccc}
1&-1&-1&-1&1&1&0\\
0&-1&0&0&1&1&0\\
0&0&-1&0&1&0&1\\
0&0&0&-1&0&1&1\\
0&0&0&0&1&0&0\\
0&0&0&0&0&1&0\\
0&0&0&0&0&0&1
\end{array}\!\right).
\end{equation*}
    The linear system we want to solve in order to obtain the multiplicities is the system in \eqref{eq: index theorem integral}. Now we note that the rank of the $D$-module annihilating the integral is given by the Euler characteristic $\chi(Y\cap\varphi^{-1}(z))$, for $z$ not in the singular locus, which is two. The index on each of the seven strata is given by the Euler characteristic $2-\chi(Y\cap\varphi^{-1}(z))$ for a generic $z$ in each strata, this gives the vector of indices $(0,0,0,0,1,1,1)^T$. Plugging the rank, index and Euler obstruction matrix into \eqref{eq: index theorem integral} and rearranging the terms we find the linear system
    \begin{equation}
        \begin{pmatrix}
            2\\2\\2\\2\\1\\1\\1
        \end{pmatrix}=\left(\!\begin{array}{ccccccc}
1&-1&-1&-1&1&1&0\\
0&-1&0&0&1&1&0\\
0&0&-1&0&1&0&1\\
0&0&0&-1&0&1&1\\
0&0&0&0&1&0&0\\
0&0&0&0&0&1&0\\
0&0&0&0&0&0&1
\end{array}\!\right)\begin{pmatrix}
    m_1\\m_2\\m_3\\m_4\\m_5\\m_6\\m_7
\end{pmatrix}\iff\begin{pmatrix}
    m_1\\m_2\\m_3\\m_4\\m_5\\m_6\\m_7
\end{pmatrix}=\begin{pmatrix}
    0\\0\\0\\0\\1\\1\\1
\end{pmatrix}.
    \end{equation}
    The characteristic cycle of this $D$-module is therefore
    \begin{equation*}
        2\Con(0)+\Con(z_3)+\Con(z_1)+\Con(z_2^2-4z_1z_3).
    \end{equation*}
    That these were the only components contributing to the cycle we already knew from the calculation of the characteristic variety. This cycle is of course the same cycle as obtained from the GKZ theory.
\end{example}
\subsection{Feynman integrals}
We conclude the paper with two examples of Feynman integrals from physics.
\begin{example}[Bubble]\label{ex: bubble} 
    The example from the Introduction is on of the simplest Feynman integrals and we revisit this example here just to verify the result. We do this by again noting that this integral in fact fits within the GKZ framework. The hypergoemetric system $H_A(\beta)$ annihilating this integral is generated by
    \begin{equation*}
        A=\begin{pmatrix}
            1&1&1&1&1\\
            1&0&1&2&0\\
            0&1&1&0&2
        \end{pmatrix},\quad \beta=\begin{pmatrix}
            -D/2\\-\nu_1\\-\nu_2.
        \end{pmatrix}
    \end{equation*}
    This system has characteristic cycle
    \begin{equation*}
        CC(H_A(\beta))=3\Con(0)+\Con(z_3)+\Con((z_1+z_2-z_3)^2-4z_1z_2)+\Con(z_1)+\Con(z_2)
    \end{equation*}
    and the support of this cycle is the same characteristic variety as obtained in the introduction.

\end{example}
\begin{example}[Parachute] In this example we consider the Feynman integral specified by the diagram below:
    \begin{center}
\begin{tikzpicture}[baseline=-\the\dimexpr\fontdimen22\textfont2\relax]
            \begin{feynman}
            \vertex (v1) at ( -\xs, 0);
            \vertex (v2) at ( {.5*\xs}, {-0.86602540378*\xs});
            \vertex (v3) at ( {.5*\xs}, {0.86602540378*\xs});

            \vertex (i1) at ( {-1.9*\xs}, 0){\(p_3\)};
            \vertex (i2) at ( {1.9*\xs/2}, {-1.9*\xs*0.86602540378}){\(p_2\)};
            \vertex (i3) at ( {1.9*\xs/2}, { 1.9*\xs*0.86602540378}){\(p_1\)};
    \diagram*{
        (v1)[dot]--[edge label'=\(m_4\)](v2)--[edge label'=\(m_1\)](v3)--[edge label'=\(m_3\)](v1),(v2)--[half right, edge label'=\(m_2\)](v3),
        (i1) -- [fermion](v1), (i2) -- [fermion](v2),(i3) -- [fermion](v3);
    };
    \end{feynman}
    \end{tikzpicture}
\end{center} 
where we fix the masses $m_1^2=1$, $m_2^2=m_3^2=0$ and $m_4^2=2$, along with two of the three external momenta $p_1^2=-1$, and $p_2^2=0$. The homogenized integrand is
\begin{align*}
    G=&x_0x_1x_2x_3x_4( -p_3^2x_1x_3x_4-p_3^2x_2x_3x_4+x_0x_1x_2+x_1^2x_2+x_0x_1x_3+x_1^2x_3+x_0x_2x_3\\
    &+2x_1x_2x_3+x_0x_1x_4+x_1^2x_4+x_0x_2x_4+3x_1x_2x_4+2x
      _1x_3x_4+2x_2x_3x_4+2x_1x_4^2+2x_2x_4^2)
\end{align*}
where we treat $p_3^2$ as an irreducible complex variable. Taking $Y=\bV(G)$ and computing a minimal Whitney stratification we obtain 72 strata, all with non-zero multiplicity:
\begin{align*}
 \{&1, 1, 3, 4, 5, 1, 1, 1, 1, 3, 1, 1, 33, 1, 1, 2, 8, 3, 7, 1, 1, 12, 2, 1,
      1, 3, 6, 3, 1, 1, 1, 3, 1, 3, 2,\\
      &1, 1, 1, 1, 2, 2, 3, 1, 3, 1, 1, 1, 1, 1,
      1, 1, 1, 1, 1, 1, 1, 1, 1, 1, 1, 1, 1, 2, 1, 1, 1, 1, 1, 1, 1, 1, 1\}.
\end{align*}
The characteristic variety of the $D$-module annihilating this integral is generated by the zero-section and the unions of the conormal varieties of the five points $\{p_3^2-2,\,p_3^2-1,\,p_3^2,\,p_3^2+1,\,p_3^2+2\}$ which also constitutes the singular locus. This coincides with \cite{Helmer:2024wax}. The rank of this $D$-module is 8 and the characteristic cycle is
\begin{equation*}
    8\Con(0)+2\Con(p_3^2-2)+\Con(p_3^2-1)+3\Con(p_3^2)+2\Con(p_3^2+1)+\Con(p_3^2+2).
\end{equation*}

\end{example}


\section*{Acknowledgments}
The authors would like to thank David Massey for a very helpful and informative correspondence.

FT is funded by the Royal Society grant number URF\textbackslash R1\textbackslash 201473. For the purpose of Open Access, the author has applied a CC BY public copyright licence to any Author Accepted Manuscript (AAM) version arising from this submission.

MH is supported by the Air Force Office of Scientific Research (AFOSR) under award
number FA9550-22-1-0462, managed by Dr.~Frederick Leve, and by the Royal Society under grant RSWF\textbackslash R2\textbackslash 242006, and would like to gratefully acknowledge this support. 
\bibliographystyle{JHEP}
\bibliography{library}

\providecommand{\href}[2]{#2}\begingroup\raggedright\begin{thebibliography}{10}

\bibitem{Weinzierl:2022eaz}
S.~Weinzierl, \emph{{Feynman Integrals}}.
\newblock 1, 2022,
  \href{https://doi.org/10.1007/978-3-030-99558-4}{10.1007/978-3-030-99558-4}.

\bibitem{Nilsson2013}
L.~Nilsson and M.~Passare, \emph{Mellin transforms of multivariate rational
  functions}, \href{https://doi.org/10.1007/s12220-011-9235-7}{\emph{J. Geom.
  Anal.} {\bfseries 23} (2013) 24--46},
  [\href{https://arxiv.org/abs/1010.5060}{{\ttfamily 1010.5060}}].

\bibitem{Berkesch2014}
C.~Berkesch, J.~Forsg{\aa}rd and M.~Passare, \emph{Euler-{M}ellin integrals and
  {$A$}-hypergeometric functions},
  \href{https://doi.org/10.1307/mmj/1395234361}{\emph{Michigan Math. J.}
  {\bfseries 63} (2014) 101--123},
  [\href{https://arxiv.org/abs/1103.6273}{{\ttfamily 1103.6273}}].

\bibitem{Lee2013}
R.~N. Lee and A.~A. Pomeransky, \emph{{Critical points and number of master
  integrals}}, \href{https://doi.org/10.1007/JHEP11(2013)165}{\emph{JHEP}
  {\bfseries 11} (2013) 165},
  [\href{https://arxiv.org/abs/1308.6676}{{\ttfamily 1308.6676}}].

\bibitem{Landau:1959fi}
L.~D. Landau, \emph{{On analytic properties of vertex parts in quantum field
  theory}},
  \href{https://doi.org/10.1016/B978-0-08-010586-4.50103-6}{\emph{Nucl. Phys.}
  {\bfseries 13} (1959) 181--192}.

\bibitem{Eden:1966dnq}
R.~J. Eden, P.~V. Landshoff, D.~I. Olive and J.~C. Polkinghorne, \emph{{The
  analytic S-matrix}}.
\newblock Cambridge Univ. Press, Cambridge, 1966.

\bibitem{Nakanishi1971}
N.~Nakanishi, \emph{Graph theory and Feynman integrals.}
\newblock Mathematics and its applications: 11. Gordon and Breach, 1971.

\bibitem{Henn:2013pwa}
J.~M. Henn, \emph{{Multiloop integrals in dimensional regularization made
  simple}}, \href{https://doi.org/10.1103/PhysRevLett.110.251601}{\emph{Phys.
  Rev. Lett.} {\bfseries 110} (2013) 251601},
  [\href{https://arxiv.org/abs/1304.1806}{{\ttfamily 1304.1806}}].

\bibitem{Dlapa:2023cvx}
C.~Dlapa, M.~Helmer, G.~Papathanasiou and F.~Tellander, \emph{{Symbol alphabets
  from the Landau singular locus}},
  \href{https://doi.org/10.1007/JHEP10(2023)161}{\emph{JHEP} {\bfseries 10}
  (2023) 161}, [\href{https://arxiv.org/abs/2304.02629}{{\ttfamily
  2304.02629}}].

\bibitem{Goncharov:2010jf}
A.~B. Goncharov, M.~Spradlin, C.~Vergu and A.~Volovich, \emph{{Classical
  Polylogarithms for Amplitudes and Wilson Loops}},
  \href{https://doi.org/10.1103/PhysRevLett.105.151605}{\emph{Phys.Rev.Lett.}
  {\bfseries 105} (Oct., 2010) 151605},
  [\href{https://arxiv.org/abs/1006.5703}{{\ttfamily 1006.5703}}].

\bibitem{Gardi:2022khw}
E.~Gardi, F.~Herzog, S.~Jones, Y.~Ma and J.~Schlenk, \emph{{The on-shell
  expansion: from Landau equations to the Newton polytope}},
  \href{https://doi.org/10.1007/JHEP07(2023)197}{\emph{JHEP} {\bfseries 07}
  (2023) 197}, [\href{https://arxiv.org/abs/2211.14845}{{\ttfamily
  2211.14845}}].

\bibitem{Brown:2009ta}
F.~C.~S. Brown, \emph{{On the periods of some Feynman integrals}},
  \href{https://arxiv.org/abs/0910.0114}{{\ttfamily 0910.0114}}.

\bibitem{Panzer:2014caa}
E.~Panzer, \emph{{Algorithms for the symbolic integration of hyperlogarithms
  with applications to Feynman integrals}},
  \href{https://doi.org/10.1016/j.cpc.2014.10.019}{\emph{Comput. Phys. Commun.}
  {\bfseries 188} (2015) 148--166},
  [\href{https://arxiv.org/abs/1403.3385}{{\ttfamily 1403.3385}}].

\bibitem{Klausen:2021yrt}
R.~P. Klausen, \emph{{Kinematic singularities of Feynman integrals and
  principal A-determinants}},
  \href{https://doi.org/10.1007/JHEP02(2022)004}{\emph{JHEP} {\bfseries 02}
  (2022) 004}, [\href{https://arxiv.org/abs/2109.07584}{{\ttfamily
  2109.07584}}].

\bibitem{Mizera:2021icv}
S.~Mizera and S.~Telen, \emph{{Landau discriminants}},
  \href{https://doi.org/10.1007/JHEP08(2022)200}{\emph{JHEP} {\bfseries 08}
  (2022) 200}, [\href{https://arxiv.org/abs/2109.08036}{{\ttfamily
  2109.08036}}].

\bibitem{Berghoff:2022mqu}
M.~Berghoff and E.~Panzer, \emph{{Hierarchies in relative Picard-Lefschetz
  theory}},  \href{https://arxiv.org/abs/2212.06661}{{\ttfamily 2212.06661}}.

\bibitem{Fevola:2023short}
C.~Fevola, S.~Mizera and S.~Telen, \emph{{Landau Singularities Revisited:
  Computational Algebraic Geometry for Feynman Integrals}},
  \href{https://doi.org/10.1103/PhysRevLett.132.101601}{\emph{Phys. Rev. Lett.}
  {\bfseries 132} (2024) 101601},
  [\href{https://arxiv.org/abs/2311.14669}{{\ttfamily 2311.14669}}].

\bibitem{Helmer:2024wax}
M.~Helmer, G.~Papathanasiou and F.~Tellander, \emph{{Landau Singularities from
  Whitney Stratifications}},
  \href{https://arxiv.org/abs/2402.14787}{{\ttfamily 2402.14787}}.

\bibitem{Caron-Huot:2024brh}
S.~Caron-Huot, M.~Correia and M.~Giroux, \emph{{Recursive Landau Analysis}},
  \href{https://arxiv.org/abs/2406.05241}{{\ttfamily 2406.05241}}.

\bibitem{Hannesdottir:2024hke}
H.~S. Hannesdottir, A.~J. McLeod, M.~D. Schwartz and C.~Vergu,
  \emph{{Applications of the Landau bootstrap}},
  \href{https://doi.org/10.1103/PhysRevD.111.085003}{\emph{Phys. Rev. D}
  {\bfseries 111} (2025) 085003},
  [\href{https://arxiv.org/abs/2410.02424}{{\ttfamily 2410.02424}}].

\bibitem{Fevola:2024acq}
C.~Fevola and S.-J. Matsubara-Heo, \emph{{Euler Discriminant of Complements of
  Hyperplanes}},  \href{https://arxiv.org/abs/2411.19696}{{\ttfamily
  2411.19696}}.

\bibitem{helmer2023effective}
M.~Helmer, A.~Leykin and V.~Nanda, \emph{Effective {W}hitney {S}tratification
  of {R}eal {A}lgebraic varieties},
  \href{https://arxiv.org/abs/2307.05427v3}{{\ttfamily 2307.05427v3}}.

\bibitem{Harris_2019}
C.~Harris and M.~Helmer, \emph{Segre class computation and practical
  applications}, \href{https://doi.org/10.1090/mcom/3448}{\emph{Math. Comp.}
  {\bfseries 89} (May, 2019) 465–491},
  [\href{https://arxiv.org/abs/1806.07408}{{\ttfamily 1806.07408}}].

\bibitem{M2}
D.~R. Grayson and M.~E. Stillman, ``{Macaulay2, a software system for research
  in algebraic geometry}.'' Available at
  \url{http://www.math.uiuc.edu/Macaulay2}.

\bibitem{Whitney1965}
H.~Whitney, \emph{Tangents to an analytic variety},
  \href{https://doi.org/10.2307/1970400}{\emph{Ann. of Math.} {\bfseries 81}
  (1965) 496--549}.

\bibitem{hnFOCM}
M.~Helmer and V.~Nanda, \emph{Conormal spaces and {W}hitney stratifications},
  \href{https://doi.org/10.1007/s10208-022-09574-8}{\emph{Found. Comput. Math.}
  {\bfseries 23} (2023) 1745--1780}.

\bibitem{eisenbud2013commutative}
D.~Eisenbud, \emph{Commutative algebra: with a view toward algebraic geometry},
  vol.~150.
\newblock Springer Science \& Business Media, 2013.

\bibitem{teissier1981varietes}
B.~Teissier, \emph{Vari\'et\'es polaires. {II}. {M}ultiplicit\'es polaires,
  sections planes, et conditions de {W}hitney},  in \emph{Algebraic geometry
  ({L}a {R}\'abida, 1981)}, vol.~961 of \emph{Lecture Notes in Math.},
  pp.~314--491.
\newblock Springer, Berlin, 1982.
\newblock \href{https://doi.org/10.1007/BFb0071291}{DOI}.

\bibitem{MacPherson1974}
R.~D. MacPherson, \emph{Chern classes for singular algebraic varieties},
  \href{https://doi.org/10.2307/1971080}{\emph{Ann. of Math.} {\bfseries 100}
  (1974) 423--432}.

\bibitem{Kashiwara1973}
M.~Kashiwara, \emph{Index theorem for a maximally overdetermined system of
  linear differential equations}, {\emph{Proc. Japan Acad.} {\bfseries 49}
  (1973) 803--804}.

\bibitem{SMTbook}
M.~Goresky and R.~MacPherson, \emph{Stratified Morse Theory}.
\newblock Springer-Verlag, 1988.

\bibitem{BDK1981}
J.-L. Brylinski, A.~S. Dubson and M.~Kashiwara, \emph{Formule de l'indice pour
  modules holonomes et obstruction d'{E}uler locale}, {\emph{C. R. Acad. Sci.
  Paris S\'er. I Math.} {\bfseries 293} (1981) 573--576}.

\bibitem{Dubson1984}
A.~S. Dubson, \emph{Formule pour l'indice des complexes constructibles et des
  {M}odules holonomes}, {\emph{C. R. Acad. Sci. Paris S\'er. I Math.}
  {\bfseries 298} (1984) 113--116}.

\bibitem{trang1981varietes}
L.~D. Tr\'ang and B.~Teissier, \emph{Vari\'et\'es polaires locales et classes
  de {C}hern des vari\'et\'es singuli\`eres},
  \href{https://doi.org/10.2307/1971299}{\emph{Ann. of Math. (2)} {\bfseries
  114} (1981) 457--491}.

\bibitem{BMM1994}
J.~Brian\c{c}on, P.~Maisonobe and M.~Merle, \emph{Localisation de syst\`emes
  diff\'erentiels, stratifications de {W}hitney et condition de {T}hom},
  \href{https://doi.org/10.1007/BF01232255}{\emph{Invent. Math.} {\bfseries
  117} (1994) 531--550}.

\bibitem{SST}
M.~Saito, B.~Sturmfels and N.~Takayama, \emph{Gr{\"o}bner deformations of
  hypergeometric differential equations}, vol.~6.
\newblock Springer Science \& Business Media, 2013.

\bibitem{Oaku1994}
T.~\={O}aku, \emph{Computation of the characteristic variety and the singular
  locus of a system of differential equations with polynomial coefficients},
  \href{https://doi.org/10.1007/BF03167233}{\emph{Japan J. Indust. Appl. Math.}
  {\bfseries 11} (1994) 485--497}.

\bibitem{Kashiwara2003}
M.~Kashiwara, \emph{{$D$}-modules and microlocal calculus}, vol.~217 of
  \emph{Translations of Mathematical Monographs}.
\newblock American Mathematical Society, Providence, RI, 2003,
  \href{https://doi.org/10.1090/mmono/217}{10.1090/mmono/217}.

\bibitem{Bjork1993}
J.-E. Bj\"ork, \emph{Analytic {$\mathscr{D}$}-modules and applications},
  vol.~247 of \emph{Mathematics and its Applications}.
\newblock Kluwer Academic Publishers Group, Dordrecht, 1993,
  \href{https://doi.org/10.1007/978-94-017-0717-6}{10.1007/978-94-017-0717-6}.

\bibitem{Treves2022}
F.~Treves, \emph{Analytic partial differential equations}, vol.~359 of
  \emph{Grundlehren der mathematischen Wissenschaften [Fundamental Principles
  of Mathematical Sciences]}.
\newblock Springer, Cham, 2022,
  \href{https://doi.org/10.1007/978-3-030-94055-3}{10.1007/978-3-030-94055-3}.

\bibitem{HormanderVolI}
L.~H\"{o}rmander, \emph{The analysis of linear partial differential operators.
  {I}}.
\newblock Classics in Mathematics. Springer-Verlag, Berlin, 2003.

\bibitem{Bernstein1972}
I.~N. Bern\v{s}te\u{\i}n, \emph{Analytic continuation of generalized functions
  with respect to a parameter}, {\emph{Funkcional. Anal. i Prilo\v{z}en.}
  {\bfseries 6} (1972) 26--40}.

\bibitem{KashiwaraKawai1981}
M.~Kashiwara and T.~Kawai, \emph{On holonomic systems of microdifferential
  equations. {III}. {S}ystems with regular singularities},
  \href{https://doi.org/10.2977/prims/1195184396}{\emph{Publ. Res. Inst. Math.
  Sci.} {\bfseries 17} (1981) 813--979}.

\bibitem{kashiwara1983systems}
M.~Kashiwara and J.-L. Brylinski, \emph{Systems of microdifferential
  equations}, vol.~34.
\newblock Birkh{\"a}user Boston, 1983.

\bibitem{GKZ}
I.~M. Gelfand, M.~Kapranov and A.~Zelevinsky, \emph{Discriminants, Resultants,
  and Multidimensional Determinants}.
\newblock Springer Science \& Business Media, 1994.

\bibitem{Gelfand1989}
I.~M. Gel'fand, A.~V. Zelevinski\u{\i} and M.~M. Kapranov, \emph{Hypergeometric
  functions and toric varieties},
  \href{https://doi.org/10.1007/BF01078777}{\emph{Funktsional. Anal. i
  Prilozhen.} {\bfseries 23} (1989) 12--26}.

\bibitem{le1983varietes}
D.~L{\^e} and Z.~Mebkhout, \emph{Vari{\'e}t{\'e}s caract{\'e}ristiques et
  vari{\'e}t{\'e}s polaires}, {\emph{CR Acad. Sci. Paris S{\'e}r. I Math}
  {\bfseries 296} (1983) 129--132}.

\bibitem{Sabbah1987}
C.~Sabbah, \emph{Proximit\'e{} \'evanescente. {II}. \'equations fonctionnelles
  pour plusieurs fonctions analytiques}, {\emph{Compositio Math.} {\bfseries
  64} (1987) 213--241}.

\bibitem{Gyoja1993}
A.~Gyoja, \emph{Bernstein-{S}ato's polynomial for several analytic functions},
  \href{https://doi.org/10.1215/kjm/1250519266}{\emph{J. Math. Kyoto Univ.}
  {\bfseries 33} (1993) 399--411}.

\bibitem{Dimca2004sheaves}
A.~Dimca, \emph{Sheaves in topology}.
\newblock Universitext. Springer-Verlag, Berlin, 2004,
  \href{https://doi.org/10.1007/978-3-642-18868-8}{10.1007/978-3-642-18868-8}.

\bibitem{Laurentiu2022}
L.~G. Maxim and J.~Sch\"urmann, \emph{Constructible sheaf complexes in complex
  geometry and applications},  in \emph{Handbook of geometry and topology of
  singularities {III}}, pp.~679--791.
\newblock Springer, Cham, 2022.
\newblock \href{https://doi.org/10.1007/978-3-030-95760-5\_10}{DOI}.

\bibitem{kashiwara-schapira1}
M.~Kashiwara and P.~Schapira, \emph{Sheaves on Manifolds}.
\newblock No.~292 in Grundlehren der math. Wiss. Springer-Verlag, 1990.

\bibitem{Kashiwara1984}
M.~Kashiwara, \emph{The {R}iemann-{H}ilbert problem for holonomic systems},
  \href{https://doi.org/10.2977/prims/1195181610}{\emph{Publ. Res. Inst. Math.
  Sci.} {\bfseries 20} (1984) 319--365}.

\bibitem{Mebkhout1984}
Z.~Mebkhout, \emph{Une \'equivalence de cat\'egories}, {\emph{Compositio Math.}
  {\bfseries 51} (1984) 51--62}.

\bibitem{Mebkhout1984b}
Z.~Mebkhout, \emph{Une autre \'equivalence de cat\'egories}, {\emph{Compositio
  Math.} {\bfseries 51} (1984) 63--88}.

\bibitem{SGA7}
P.~Deligne and N.~Katz, \emph{Groupes de monodromie en g\'eom\'etrie
  alg\'ebrique. {II}}, vol.~340 of \emph{Lecture Notes in Mathematics}.
\newblock Springer-Verlag, Berlin-New York, 1973.

\bibitem{kashiwara1975maximally}
M.~Kashiwara, \emph{On the maximally overdetermined system of linear
  differential equations. {I}},
  \href{https://doi.org/10.2977/prims/1195192011}{\emph{Publ. Res. Inst. Math.
  Sci.} {\bfseries 10} (1974/75) 563--579}.

\bibitem{hotta2007d}
R.~Hotta and T.~Tanisaki, \emph{D-modules, perverse sheaves, and representation
  theory}, vol.~236.
\newblock Springer Science \& Business Media, 2007.

\bibitem{Mebkhout1977}
Z.~Mebkhout, \emph{Local cohomology of analytic spaces},
  \href{https://doi.org/10.2977/prims/1195196607}{\emph{Publ. Res. Inst. Math.
  Sci.} {\bfseries 12} (1977) 247--256}.

\bibitem{WhitStratM2}
M.~Helmer, ``Whitney{S}tratifications {M}acaulay2 package.''
  \url{https://macaulay2.com/doc/Macaulay2/share/doc/Macaulay2/WhitneyStratifications/html/index.html},
  2022--2025.

\bibitem{SegreM2}
C.~Harris and M.~Helmer, ``{S}egre{C}lasses {M}acaulay2 package.''
  \url{https://macaulay2.com/doc/Macaulay2/share/doc/Macaulay2/SegreClasses/html/index.html},
  2018--2025.

\bibitem{Briancon1975}
J.~Brian\c{c}on and J.-P. Speder, \emph{La trivialit\'e{} topologique
  n'implique pas les conditions de {W}hitney}, {\emph{C. R. Acad. Sci. Paris
  S\'er. A-B} {\bfseries 280} (1975) Aiii, A365--A367}.

\bibitem{Kallstrom1989}
R.~K\"allstr\"om, \emph{Meromorphic extensions of regular holonomic
  distributions}, \href{https://doi.org/10.1007/BF02124335}{\emph{Comm. Math.
  Phys.} {\bfseries 126} (1989) 157--166}.

\bibitem{Kashiwara1987}
M.~Kashiwara, \emph{Regular holonomic {$D$}-modules and distributions on
  complex manifolds},  in \emph{Complex analytic singularities}, vol.~8 of
  \emph{Adv. Stud. Pure Math.}, pp.~199--206.
\newblock North-Holland, Amsterdam, 1987.
\newblock \href{https://doi.org/10.2969/aspm/00810199}{DOI}.

\bibitem{kashiwara1979characteristic}
M.~Kashiwara and T.~Kawai, \emph{On the characteristic variety of a holonomic
  system with regular singularities},
  \href{https://doi.org/10.1016/0001-8708(79)90055-0}{\emph{Adv. in Math.}
  {\bfseries 34} (1979) 163--184}.

\bibitem{KashiwaraKawai1979}
M.~Kashiwara and T.~Kawai, \emph{On holonomic systems for {$\Pi
  \sp{N}\sb{l=1}(f\sb{l}+(\surd \ -1)0)\sp{\lambda \sb{l}}$}},
  \href{https://doi.org/10.2977/prims/1195188184}{\emph{Publ. Res. Inst. Math.
  Sci.} {\bfseries 15} (1979) 551--575}.

\bibitem{MumfordOda}
D.~Mumford and T.~Oda, \emph{Algebraic geometry {II}},  2015.

\bibitem{stasica2002effective}
A.~Stasica, \emph{An effective description of the {J}elonek set},
  \href{https://doi.org/10.1016/S0022-4049(01)00063-9}{\emph{J. Pure Appl.
  Algebra} {\bfseries 169} (2002) 321--326}.

\bibitem{Huh2013}
J.~Huh, \emph{The maximum likelihood degree of a very affine variety},
  \href{https://doi.org/10.1112/S0010437X13007057}{\emph{Compos. Math.}
  {\bfseries 149} (2013) 1245--1266}.

\bibitem{helmer2016proj}
M.~Helmer, \emph{Algorithms to compute the topological {E}uler characteristic,
  {C}hern-{S}chwartz-{M}ac{P}herson class and {S}egre class of projective
  varieties}, \href{https://doi.org/10.1016/j.jsc.2015.03.007}{\emph{J.
  Symbolic Comput.} {\bfseries 73} (2016) 120--138}.

\bibitem{charClassM2}
M.~Helmer and C.~Jost, ``{C}haracteristic{C}lasses {M}acaulay2 package.''
  \url{https://www.macaulay2.com/doc/Macaulay2/share/doc/Macaulay2/CharacteristicClasses/html/index.html},
  2015--2025.

\end{thebibliography}\endgroup
\end{document}